\documentclass[12pt,a4paper]{article}
\usepackage{enumitem}
\usepackage{geometry}
\usepackage{amsmath,amsthm}
\usepackage{algorithm}
\usepackage{algorithmic}
\usepackage{makecell}
\usepackage{multirow}
\usepackage{tikz}
\usepackage{setspace}
\usepackage{fullpage} 
\makeatletter
\newenvironment{breakablealgorithm}
  {
   \begin{center}
     \refstepcounter{algorithm}
     \hrule height.8pt depth0pt \kern2pt
     \renewcommand{\caption}[2][\relax]{
       {\raggedright\textbf{\ALG@name~\thealgorithm} ##2\par}%
       \ifx\relax##1\relax 
         \addcontentsline{loa}{algorithm}{\protect\numberline{\thealgorithm}##2}%
       \else 
         \addcontentsline{loa}{algorithm}{\protect\numberline{\thealgorithm}##1}%
       \fi
       \kern2pt\hrule\kern2pt
     }
  }{
     \kern2pt\hrule\relax
   \end{center}
  }
\makeatother
\allowdisplaybreaks[4]
\usepackage{multicol}
\usepackage{hyperref}
\usepackage{caption}[=v1]
\usepackage{color}
\usepackage[titletoc]{appendix}
\usepackage{float}
\usepackage{bm}
\usepackage{enumerate}
\usepackage{booktabs}
\usepackage{extarrows}
\usepackage{graphicx}
\usepackage{subfigure}
\usepackage{amssymb}
\usepackage{fancyhdr}
\usepackage{authblk}
\usepackage{dsfont}
\usepackage{nonfloat}
\usepackage{xcolor}
\usepackage{tikz} 
\usetikzlibrary{arrows,shapes,chains}
\usepackage{hyperref} 
\usepackage{xfrac} 
\usepackage{longtable}
\usepackage{booktabs}
\usepackage{subfigure}
\usepackage{caption}
\usepackage[comma,authoryear]{natbib}
\newtheorem{theorem}{Theorem}[section]
\newtheorem{lemma}{Lemma}[section]

\newcommand{\comm}{\color{black}}

\begin{document}

\title{\textbf {The Exploratory Multi-Asset Mean-Variance Portfolio Selection using Reinforcement Learning}}

	\author[a]{Yu Li}
	\author[a]{Yuhan Wu}
    \author[a,b]{Shuhua Zhang\footnote{Corresponding author}}
	\affil[a]{Coordinated Innovation Center for Computable Modeling in Management Science,

	Tianjin University of Finance and Economics, Tianjin 300222, China}
	\affil[b]{Zhujiang College, South China Agricultural University, Guangzhou 510900, China}
	\renewcommand*{\Affilfont}{\footnotesize\it}

\renewcommand{\thefootnote}{}
\footnotetext{E-mail addresses: liyu@tjufe.edu.cn (Yu Li), wuyuhan@stu.tjufe.edu.cn (Yuhan Wu), szhang@tjufe.edu.cn (Shuhua Zhang)}

\date{}
\maketitle

\begin{abstract}

In this paper, 
we study
the continuous-time multi-asset mean-variance (MV) portfolio selection using a reinforcement learning (RL) algorithm, 
specifically the soft actor-critic (SAC) algorithm, 
in the time-varying financial market.
A family of Gaussian portfolio selections is derived, 
and a policy iteration process is crafted to learn the optimal {\comm exploratory} portfolio selection. 
We prove the convergence of the policy iteration process theoretically, based on which the SAC algorithm is developed.
To improve the algorithm's stability and the learning accuracy in the multi-asset scenario, we divide the model parameters that influence the optimal portfolio selection into {\comm three parts}, and learn each part progressively.
Numerical studies in the simulated and real financial markets confirm the superior performance of the proposed SAC algorithm under various criteria.


\end{abstract}

\renewcommand{\thefootnote}{}
\footnotetext{Abbreviation statement: soft actor-critic (SAC); mean-variance (MV)}


\newpage
\section{Introduction}

In financial markets, the assets are broadly categorized into two types: risky assets and riskless assets.
The portfolio selection problem is a study of seeking an allocation of wealth among these different assets.
Since the future prices of risky assets are unknown, 
investors are always looking for the portfolios which effectively balance investment return opportunities against risks. 
\cite{Markowitz1952Portfolio} lays the fundamental basis for the portfolio selection problem by modeling the prices of risky assets as random variables.
Then, the investment returns and risks are quantified by the expectation and variance of the portfolios, respectively, forming the mean-variance (MV) model.

Since the establishment of the MV model, scholars have sought to enhance its applicability in dynamic financial markets.
\cite{Merton1969Lifetime} pioneers the groundwork of dynamic stochastic processes, and \cite{Li2000Optimal} constructs a continuous-time MV model via stochastic linear-quadratic control theory.
In this continuous-time framework, the portfolio is viewed as a continuous sequence of decisions, allowing investors to adjust their allocations dynamically in response to evolving financial markets.
Subsequently, numerous related issues have been extensively investigated under the continuous-time MV framework, including mean-variance hedging \citep{Schweizer2010Mean}, {\comm ``local mean-variance efficiency''} \citep{Czichowsky2013Time}, and state-dependent risk aversion \citep{Bjoerk2014Mean}. 
Among these researches, the multi-asset portfolio selection problem holds critical importance, in which the correlations among different risky assets are considered.
In the multi-asset context, investors can diversify their wealth across these risky assets to reduce investment risk, under a given expected return.


The traditional paradigm for implementing {\comm MV portfolio} follows ``separation principle'', which separates the steps between estimation and optimization.
In the first step, model parameters are estimated from time-series data of risky asset prices using maximum likelihood estimation (MLE). 
In the second step, these estimated parameters are taken as given, and optimization of the MV model is focused on. 
However, researches have documented that this ``separation principle'' is difficult to generate good out-of-sample performance,
especially in multi-asset financial markets
\citep{Jobson1981Putting,Broadie1993Computing,DeMiguel2007Optimal, Ledoit2017Nonlinear, Lian2019Portfolio, Barroso2022Lest, Hiraki2022toolkit}.

One popular method to mitigate estimation errors in portfolio selection is the use of shrinkage estimators, which balances the low bias of sample-based estimation with the low variance of pre-specified structural models (e.g., single-index frameworks, constant-correlation matrices, or equally-weighted portfolio benchmarks).
For the expected return vector, \cite{Jorion1986Bayes} introduces the Bayes-Stein estimation, shrinking sample means toward a prior belief to reduce estimation variance.
For the covariance matrix, \cite{Ledoit2003Improved,Ledoit2004Honey,Ledoit2004well} linearly combine the sample covariance matrix with structured models (e.g., factor models) to minimize mean squared error and enhance robustness in high-dimensional scenarios.
What's more, \cite{Candelon2012Sampling} proposes a double shrinkage methodology: first applying Bayes-Stein shrinkage to the covariance matrix, then regularizing the portfolio toward an equally-weighted benchmark, thereby further reducing the sampling error in the case of small samples.
Additionally, inverse covariance matrix shrinkage offers an alternative approach, applying directly to portfolio optimization without the need to calculate the inverse. 
The conditional number regularization estimator proposed by \cite{Won2013Condition} 
and the eigenvalues regularization estimator proposed by \cite{Shi2020Improving} represent the unstructured estimation approaches for inverse covariance matrix without priori beliefs.

However, there exists inconsistency in ``separation principle'' between the parameters estimation and portfolio optimization.
Estimation is intended to minimize its prediction error, i.e., mean squared error, while optimization is to maximize the MV utility. 
In contrast, reinforcement learning (RL) algorithms avoid the inconsistency in the traditional paradigm, which learn the optimal portfolio selection directly through interactions with the financial market \citep{Fischer2018Reinforcement}.
To improve the generalization capabilities of portfolios,
\cite{Haarnoja2018Softa,Haarnoja2018Softb} firstly proposes a Soft Actor-Critic (SAC) algorithm,
and \cite{Wang2020Continuous} develops the SAC algorithm into the continuous-time single-asset MV model with a stationary financial market.
In their approach, the SAC agent learns the optimal allocation with a exploratory portfolio selection, iterating by the corresponding policy evaluation and policy improvement theorem.
Numerical results from their study indicate that the SAC algorithm has significant advantages over the traditional paradigm, including higher investment returns, lower risks, improved Sharpe ratio, reduced training time, and a faster-converging learning curve. 
This study provides a new scheme for applying RL algorithm to continuous-time MV portfolio selection problem with better performance.

In the subsequent years, the application of SAC algorithm has garnered substantial attention across diverse investment management scenarios.
For instance,
\cite{Zhu2021Optimal} conducts a paired-trading study of two risky assets under the MV model. 
\cite{Guo2022Entropy} introduces the impact of exploratory portfolio selection with learning in the mean-field game. 
\cite{Jiang2022Reinforcement} explores the wealth allocation under the Kelly criterion. 
\cite{Bender2023Entropy} considers that the risky asset prices contain jump processes. 
\cite{Aquino2023Portfolio} innovatively considers the trade-off between exploiting existing assets and exploring investment opportunities with new assets.
\cite{Dai2023Learning} studys the Merton utility maximization problem with time-consistent portfolio selections. 
Although numerous studies have expanded the SAC-based portfolio selection problem from various aspects, the majority still focuses on single-asset problems.

When it comes to the continuous-time multi-asset MV portfolio selection, the SAC algorithm shows promise but encounters challenges.
As the number of risky assets increases, the computational complexity of the SAC algorithm escalates exponentially.  
This exponential growth in complexity not only hampers the algorithm's efficiency but also limits its scalability for practical applications.
Moreover, maintaining the learning accuracy and stability of SAC algorithms in a multi-asset context is a formidable task.  
The complex interactions among multiple assets introduce additional noise and uncertainty, making it difficult to ensure that the algorithm converges stably. 
Without stable convergence, the performance and effectiveness of the learned portfolios remain dubious. 

In this paper, we aim to tackle these challenges. 
In order to enhance the learning stability and efficiency of the SAC algorithm in multi-asset portfolio selection, we decouple the learning processes.
We first separate the long-term and immediate factors that influence the portfolio selection.
When learning long-term factors, we focus on exploring stable patterns embedded in macroeconomic trends and industry prospects, avoiding the interference of immediate factors, and thus improving the stability of the learning process.
When learning immediate factors, we separate the different risky asset factors that affect the portfolio selection and learn them independently.
This allows us to improve the accuracy and efficiency of the learning process.

Numerical experiments are carried out across various simulated and real financial markets.
In the simulated settings, 
the SAC algorithm demonstrates higher precision in learning parameters compared to MLE.
Additionally, the SAC algorithm yields highly robust and stable results, highlighting its effectiveness under diverse market conditions.
When applied to real financial markets,
four portfolio selections are compared: 
the portfolio selection with our online SAC algorithm, the portfolio selection with maximum likelihood estimation (MLE), the so-called ``buy-and-hold'' portfolio selection, and the broad-market index.
And our SAC algorithm shows remarkable superiority under various criteria, including terminal wealth, Certainty-Equivalent Return (CEQ), and Sharpe Ratio (SR).

In summary, the main contributions of this paper are threefold.

{\comm
1. 
For the multi-asset portfolio selection problem, we design an online SAC algorithm to learn the optimal MV portfolio selection under the time-varying financial market.

2. We develop a policy iteration procedure including policy evaluation and policy improvement for the multi-asset MV portfolio selection, and prove the convergence of the policy iteration procedure.

3. We decouple the learning processes in the SAC algorithm, which improves learning efficiency and stability for it in the multi-asset portfolio selection problem.
}

The remainder of this paper is organized as follows.
In Section 2, we formulate the classical continuous-time multi-asset MV model and show the optimal portfolio selection of it.
Section 3 provides the exploratory formulation for the continuous-time multi-asset MV model.
We develop the policy evaluation and policy improvement theorem to learn the optimal portfolio selection iteratively,
and provide a convergence result theoretically.
In Section 4, we detail the online SAC algorithm and decouple the learning processes in it.
Numerical studies and empirical analyses are presented in Section 5 under various simulated and real financial markets.
Finally, we conclude in Section 6.
Some technical proofs are relegated to Appendices.

\section{Formulation of Problem}

Assume there is one riskless asset (bond) and $n$ risky assets (stocks) available for investment. 
Let the planning investment horizon $[0,T]$ be fixed.
The riskless asset has a constant interest rate $r$.
$\{B_t^{(1)},\dots,B_t^{(n)},0 \leqslant t \leqslant T\}$ is the standard $n$-dimensional Brownian motion defined on a filtered probability space $(\Omega,\mathcal{F},\mathbb{P};\{\mathcal{F}_t\}_{0 \leqslant t \leqslant T})$.
The price of the $i$-th risky asset is observable, whose discounted value can be governed by the stochastic differential equation:
\begin{align}\label{equ:S}
	dS^{(i)}_t = 
	(\mu^{(i)}(t)-r) S^{(i)}_t dt 
	+ \sigma^{(i)}(t) S^{(i)}_t dB^{(i)}_t,
	\qquad i=1,\dots,n,
\end{align}
where the return rate $\mu^{(i)}(t)$ and volatility $\sigma^{(i)}(t)$ are time-dependent.
$dB^{(i)}_t \cdot dB^{(j)}_t=\rho^{(ij)} dt$,
in which $\rho^{(ij)}\in[-1,1]$ is constant and describes the correlation between the return of the $i$-th and the $j$-th risky asset, $i, j\in\{1,\dots,n\}$.
Throughout this paper, we denote the excess expected return vector and the covariance matrix of $n$ risky assets by
\begin{align}\label{equ:DLD}
\mu-r=
\begin{bmatrix}
	\mu^{(1)}(t)-r &\cdots &\mu^{(n)}(t)-r
\end{bmatrix}^\top
\quad\text{and}\quad
\Sigma=DLD^\top,
\end{align}
respectively, where
\begin{align}\label{equ:L}
D=\text{diag}\{\sigma^{(1)}(t),\cdots,\sigma^{(n)}(t)\},
\qquad
L=
\begin{bmatrix}
	\rho^{(11)} &\cdots &\rho^{(1n)}\\
	\rho^{(21)} &\cdots &\rho^{(2n)}\\
	\vdots	    &\vdots &\vdots\\
	\rho^{(n1)} &\cdots &\rho^{(nn)}
\end{bmatrix}.
\end{align}

\subsection{Classical continuous-time MV model}\label{se:classical_model}

We first recall the classical continuous-time mean-variance (MV) model.
The analytic expression of the multi-asset optimal MV portfolio selection is shown in Lemma \ref{th:theta_cl}.

In financial market \eqref{equ:S}, the investor rebalances the portfolio selection dynamically with an allocation $\Theta_t= \begin{bmatrix} \theta^{(1)}_t &\cdots &\theta^{(n)}_t \end{bmatrix}^\top$, $\forall t\in[0,T]$, in which $\theta_t^{(i)}$ is the discounted amount put in the $i$-th risky asset at time $t$.
Under the self-financing condition, the discounted wealth process $W_t$ follows:
\begin{align*}
	dW_t
	= \sum_{i=1}^n \dfrac{\theta_t^{(i)}}{S_t^{(i)}} dS^{(i)}_t
	= \Theta_t^\top \Big( (\mu-r) dt + D dB_t \Big),
\end{align*}
with initial wealth $w^o>0$, where $dB_t = \begin{bmatrix} dB^{(1)}_t &\cdots &dB^{(n)}_t \end{bmatrix}^\top$.
The classical continuous-time MV model aims to consider the portfolio selection which maximizes the trade-off between the expectation and variance of terminal wealth $W_T$:
\begin{equation}\label{model_cl}
	\max_{\{\Theta_t\}} {\rm E}\Big(W_T\Big)-\gamma {\rm Var}\Big(W_T\Big),
\end{equation}
where $\gamma>0$ is the risk aversion coefficient.

Because the variance operator in \eqref{model_cl} is non-smooth, i.e.,
\begin{equation*}
	{\rm Var}_s \Big({\rm Var}_t \Big(\cdot\Big)\Big)\neq {\rm Var}_s\Big(\cdot\Big), \qquad 0\leqslant s<t\leqslant T,
\end{equation*}
the principle of dynamic programming \citep{Bellman1957Dynamic} fails.
In order to obtain the optimal MV portfolio selection,
following \cite{Zhou2000Continuous}, the classical continuous-time MV model \eqref{model_cl} is transformed into a tractable stochastic linear-quadratic problem:
\begin{equation}\label{model_cl_aux}
	\max_{\{\Theta_t\}} {\rm E}\Big(-\gamma W_T^2+\tau W_T\Big),
\end{equation}
with $\tau=1+2\gamma {\rm E}\Big(W_T^*\Big)$, where $\{W_t^*\}_{0\leqslant t \leqslant T}$ is the discounted wealth with the optimal portfolio selection.
Model \eqref{model_cl_aux} can be solved analytically, whose optimal portfolio selection $\{\Theta_t^*\}_{0\leqslant t \leqslant T}$ is shown in Lemma \ref{th:theta_cl}.

\begin{lemma}[\cite{Zhou2000Continuous}]\label{th:theta_cl}
The optimal portfolio selection of model \eqref{model_cl_aux} is given by
\begin{align}\label{equ:theta_n_cl}
\Theta_t^{*}
=( \dfrac{\tau}{2\gamma}-w)\Sigma^{-1}(\mu-r),
\qquad
\forall t\in[0,T],
\end{align}
with $\tau=e^{K(0,T)\cdot T}+2\gamma w^o$, 
where $w$ and $w^o$ are respectively the discounted amount of $t$-time wealth and initial wealth,
and
\begin{align*}
	K(0,T)
	=\dfrac{1}{T}\int_0^T (\mu-r)^\top \Sigma^{-1} (\mu-r)ds.
\end{align*}
\end{lemma}

In Lemma \ref{th:theta_cl}, the optimal portfolio selection \eqref{equ:theta_n_cl} is related to three parts of parameters: $\mu-r$, $\Sigma^{-1}$ and $K(0,T)$.
Specifically, $\mu-r\in\mathbb{R}^{n\times 1}$ and $\Sigma^{-1}\in\mathbb{R}^{n\times n}$ are time-dependent representing the excess expected return vector and the inverse covariance matrix of the $n$ risky assets, respectively.
A change in the value of an element of $\mu-r$ or $\Sigma^{-1}$ exclusively influence the allocation associated with the corresponding risky assets \citep{Best1991sensitivity}.
In contrast, $K(0, T)\in\mathbb{R}$ remains constant during the whole planning investment horizon representing the average of squared Sharpe ratio of $n$ risky assets.
As the value of $K(0,T)$ increases, the amount invested in each risky asset is increased proportionally.
We will further explain the economic implications of $K(0,T)$ in Section \ref{se:K} later.

The optimal portfolio selection \eqref{equ:theta_n_cl}
in Lemma \ref{th:theta_cl} is the well-known pre-commitment portfolio selection \citep{Zhou2000Continuous, Li2000Optimal, Wang2010Numerical},
which shows superior performance within stable economic regimes \citep{Forsyth2020Multiperiod,Vigna2020time}.
When $T\to0$, the optimal portfolio selection at time 0 degenerates into
$\frac{1}{2\gamma}\Sigma^{-1}(\mu-r)$,
which is consistent with the static portfolio selection in single-period MV model \citep{Markowitz1956Optimization, Merton1972Analytical}.

\subsection{The average profitability of risky assets}\label{se:K}

In this subsection, we take a closer look at $K(0, T)$.
We explain the economic implications and then analyze the properties for it.

For the $i$-th risky asset, $\frac{\mu^{(i)}(t)-r}{\sigma^{(i)}(t)}$ is the Sharpe ratio of it at time $t$, $i=1,\dots,n$.
If $\frac{\mu^{(i)}(t)-r}{\sigma^{(i)}(t)}\neq0$, the investor can profit from buying or shorting the risky asset. 
The further $\frac{\mu^{(i)}(t)-r}{\sigma^{(i)}(t)}$ is from zero, the more the investor can earn with a share of the risky asset.
We define the square of $\frac{\mu^{(i)}(t)-r}{\sigma^{(i)}(t)}$ as $A^{(i)}(t)$, i.e.,
\begin{align*}
	A^{(i)}(t)=\Big(\dfrac{\mu^{(i)}(t)-r}{\sigma^{(i)}(t)}\Big)^2,
\end{align*}
which represents the current profitability of the $i$-th risky asset.
And, the average of $A^{(i)}(s)$ over the planning investment horizon $[t,T]$ is defined as $K^{(i)}(t,T)$,
\begin{align*}
	K^{(i)}(t,T)=\dfrac{1}{T-t}\int_t^T A^{(i)}(s) ds,
\end{align*}
which represents the average profitability of the $i$-th risky asset from $t$ to $T$. 

Similarly, in multi-asset financial market, $(\mu-r)^\top\Sigma^{-1} (\mu-r)$ is the squared Sharpe ratio of the $n$ risky assets.
We define $A(t)$ and $K(t,T)$ as 
\begin{equation}\label{equ:K}
\begin{split}
	A(t)=(\mu-r)^\top\Sigma^{-1} (\mu-r),
	\qquad
	K(t,T)=\dfrac{1}{T-t}\int_t^T A(s) ds,
\end{split}
\end{equation}
which represent the current and average profitability of $n$ risky assets, respectively.
In fact, when the financial market is stable, the average profitability of the $n$ risky assets $K(t,T)$ can be represented by the average profitability of each risky asset $K^{(i)}(t)$, $i=1,\dots,n$. 
We summarize this property into Theorem \ref{th:K_k}.

\begin{theorem}\label{th:K_k}
When the financial market is stable, i.e., model parameters $\mu, \Sigma$ are time-independent,
we have the following relation between the average profitability of $n$ risky assets and that of each risky asset, $\forall t\in[0,T)$,
\begin{align}\label{equ:K_k}
	K(t,T)
	=
\begin{bmatrix}
	\sqrt{K^{(1)}(t,T)} &\cdots &\sqrt{K^{(n)}(t,T)}~
\end{bmatrix}
L^{-1}
\begin{bmatrix}
	\sqrt{K^{(1)}(t,T)}~\\ \vdots\\ \sqrt{K^{(n)}(t,T)}
\end{bmatrix},
\end{align}
where the correlation coefficient matrix $L$ is defined in \eqref{equ:L}.
\end{theorem}
\begin{proof}
See Appendix \ref{app:K_k}.
\end{proof}

\section{The Exploratory Portfolio Selection}\label{se:exploratory_model}

Due to the lack of information about the parameters $\mu-r$, $\Sigma^{-1}$ and $K(0,T)$, the RL agent explores the financial market with a exploratory portfolio selection.
In Section \ref{se:3.1}, following \cite{Wang2020Continuous}, we develop an exploratory formulation for the continuous-time multi-asset MV problem \eqref{model_cl}.
In Section \ref{se:multi}, we derive the optimal probability density function of the exploratory portfolio selection, whose expectation is the optimal MV portfolio selection in Lemma \ref{th:theta_cl}.
In Section \ref{se:alg_convergence}, we develop a policy iteration process to learn 
this optimal exploratory portfolio selection.

\subsection{Exploratory continuous-time MV model}\label{se:3.1}

The key idea of the exploratory formulation is to consider the randomness of the portfolio selection.
The RL agent chooses its $t$-time action (portfolio) by sampling from a multivariate probability density function $P(t,\cdot)$, which is called an exploratory portfolio selection, with the constraint
\begin{align*}
\int_{\mathbb{R}^n}P(t,\theta) d\theta=1.
\end{align*}

We first describe the discounted wealth process under the exploratory portfolio selection $P(t,\theta)$.
Let $\widetilde{W}_t$ denote the $t$-time discounted wealth.
Following \cite{Wang2019Exploration}, $\{\widetilde{W}_t\}_{0 \leqslant t \leqslant T}$ is the ``average'' of infinitely many wealth processes generated 
under the portfolios that are repeatedly sampled from the probability density function $\{P(t,\cdot)\}_{0 \leqslant t \leqslant T}$.
The discounted wealth process under the exploratory formulation is described by:
\begin{equation}\label{equ:wealth_RL}
	d\widetilde{W}_t 
	= \int_{\mathbb{R}^n}\theta^\top (\mu-r) P(t,\theta) d\theta \cdot dt 
	+\sqrt{\int_{\mathbb{R}^n}\theta^\top\Sigma\theta P(t,\theta) d\theta} \cdot d\widetilde{B}_t,
\end{equation}
in which $\{\widetilde{B}_t\}_{0 \leqslant t \leqslant T}$ is the standard one-dimensional Brownian motion.

Next, we describe the objective function in the exploratory continuous-time MV model.
To regulate the level of exploration, the information entropy $h(P(t,\cdot))$ \citep{Cover1991Elements,Mnih2016Asynchronous, Nachum2017Improving} is introduced:
\begin{align}\label{equ:entropy}
	h(P(t,\cdot)):=\int_{\mathbb{R}^n}-P(t,\theta)\ln{P(t,\theta)}d\theta.
\end{align}
More uncertainty of the exploratory portfolio selection corresponds to a larger value of information entropy.
When we are on the realm of classical continuous-time MV model, the probability density function $P(t,\cdot)$ is the Dirac measure, and the information entropy $h(P(t,\cdot))$ tends to $-\infty$.
In the exploratory formulation, we encourage the exploration and incorporate the accumulative information entropy $\mathcal{H}(P(\cdot,\cdot)):=\int_0^T h(P(t,\cdot))dt$ into the objective function of the classical continuous-time MV model \eqref{model_cl}.
In fact, the accumulative information entropy $\mathcal{H}(P(\cdot,\cdot))$ has already been used by \cite{Wang2020Continuous} and \cite{Dai2023Learninga} to regularize exploration for a continuous-time single-asset MV portfolio selection problem.
Then, the entropy-regularized optimization problem 
for the continuous-time multi-asset MV model
is formulated as:
\begin{equation}\label{model_rl}
	\max_{\{P(t,\cdot)\}} {\rm E}\Big(\widetilde{W}_T\Big)-\gamma {\rm Var}\Big(\widetilde{W}_T\Big)+\lambda \mathcal{H}(P(\cdot,\cdot)),
\end{equation}
where $\lambda~(\lambda>0)$ is the exploration weight,
and 
the discounted terminal wealth $\widetilde{W}_T$
is defined in \eqref{equ:wealth_RL}
under the exploratory portfolio selection 
$\{P(t,\cdot)\}_{0 \leqslant t \leqslant T}$.

\subsection{The gaussian exploration}\label{se:multi}

In order to solve the exploratory continuous-time MV model \eqref{model_rl}, the stochastic linear-quadratic optimal control model \eqref{model_cl_aux} is also transformed into the exploratory formulation:
\begin{equation}\label{model_rl_aux}
	\max_{\{P(t,\cdot)\}} {\rm E}\Big(-\gamma \widetilde{W}_T^2+\tau \widetilde{W}_T\Big)
	+\lambda \mathcal{H}(P(\cdot,\cdot))
\end{equation}
with $\tau=1+2\gamma {\rm E}\Big(\widetilde{W}_T^*\Big)$,
where $\{\widetilde{W}_t\}_{0\leqslant t \leqslant T}$ subjects to the process \eqref{equ:wealth_RL}
and $\widetilde{W}_T^*$ is the discounted wealth at terminal time $T$ with the optimal exploratory portfolio selection.
Now, we prove the equivalence of problem \eqref{model_rl} and \eqref{model_rl_aux} in Lemma \ref{le:rl_aux}.

\begin{lemma}\label{le:rl_aux}
For $t\in[0,T]$, suppose $P^*(t,\cdot)$ is the optimal exploratory portfolio selection for original problem \eqref{model_rl}. 
Then, $P^*(t,\cdot)$ is also optimal for auxiliary problem \eqref{model_rl_aux} with $\tau=1+2\gamma {\rm E}\Big(\widetilde{W}_T^*\Big)$.
\end{lemma}

\begin{proof}
See Appendix \ref{app:rl_aux}.
\end{proof}

According to Lemma \ref{le:rl_aux}, any optimal solution of model \eqref{model_rl} can be found via solving the stochastic linear-quadratic model \eqref{model_rl_aux}.
Hence, in the following of this paper, we focus on model \eqref{model_rl_aux}.
For $\forall(t,w)\in[0,T]\times \mathbb{R}$, we define the optimal value function
\begin{equation}\label{model_rl_aux_subproblem}
\begin{split}
	V^{*}(t,w)&=\max_{\{P(s,\cdot)\}}{\rm E}\Big(-\gamma \widetilde{W}_T^2+\tau \widetilde{W}_T\Big)
	+\lambda\int_t^T h(P(s,\cdot))ds
\end{split}
\end{equation}
with $h(P(s,\cdot))$, $s\in[t,T]$, defined in \eqref{equ:entropy}.
Following the principle of dynamic programming, we deduce that $V^*(t,w)$ satisfies the Hamilton-Jacobi-Bellman (HJB) equation
\begin{equation}\label{equ:HJB}
\begin{split}
  -\dfrac{\partial V^*}{\partial t}(t,w)
	=&\max_{P(t,\cdot)}\Big\{ \lambda  h(P(t,\cdot))\\
	&+\dfrac{\partial V^*}{\partial w}(t,w)\int_{\mathbb{R}^n}\theta^\top (\mu-r) P(t,\theta) d\theta
  +\dfrac{1}{2}\dfrac{\partial^2 V^*}{\partial w^2}(t,w)\int_{\mathbb{R}^n}\theta^\top\Sigma\theta P(t,\theta) d\theta\Big\}
\end{split}
\end{equation}
with the terminal condition $V^*(T,w)=-\gamma w^2+\tau w$.
Then, applying the high dimensional Euler-Lagrange equation \citep{Liberzon2012Calculus} to HJB equation \eqref{equ:HJB}, the optimal exploratory portfolio selection $P^*(t,\theta)$ can be obtained in Theorem \ref{th:v_theta_n}.

\begin{theorem}\label{th:v_theta_n}
For $\forall(t,w) \in [0,T] \times \mathbb{R}$, the optimal exploratory portfolio selection of model \eqref{model_rl_aux} is Gaussian, whose density function is
\begin{equation}\label{equ:theta_n}
	P^*(t,\cdot)
	=\mathcal{N}\Big((\dfrac{\tau}{2\gamma}-w)\Sigma^{-1}(\mu-r),\dfrac{\lambda}{2}\dfrac{e^{K(t,T)\cdot (T-t)}}{\gamma}\Sigma^{-1}\Big).
\end{equation}
The corresponding optimal value function is given by
\begin{equation}\label{equ:V_n}
\begin{split}
	V^*(t,w)
	&=-\gamma e^{-K(t,T)\cdot (T-t)}\left(w - \frac{\tau}{2\gamma}\right)^2
	+\frac{\tau^2}{4\gamma}\\
	&+\frac{\lambda n}{2}\int_{t}^{T}\left[\ln\left(\frac{\pi\lambda}{\gamma}\right)+\frac{1}{n}\ln(|\Sigma^{-1}|)+K(s,T)\cdot(T-s)\right]ds
\end{split}
\end{equation}
with $K(t,T)$ and $K(s,T)$ defined in \eqref{equ:K}.
Moreover, $\tau=e^{K(0,T)\cdot T}+2\gamma w^o$.
\end{theorem}
\begin{proof}
	See Appendix \ref{app:v_theta_n}.
\end{proof}


\subsection{Policy evaluation and policy improvement procedure}\label{se:alg_convergence}

In this section, 
we employ a policy iteration procedure to learn the optimal exploratory portfolio selection \eqref{equ:theta_n}.
A policy iteration procedure usually consists of two circularly ongoing steps:
policy evaluation and policy improvement \citep{Sutton2018Reinforcement}.
The former provides an estimated value function for the current policy, whereas the latter updates the current policy in the right direction to improve the value function.
In this subsection, we first develop the policy evaluation and policy improvement theorem for the multi-asset exploratory MV portfolio selection with time-varying financial markets,
and then present a convergence analysis for it.

A raw indicator for evaluating the exploratory portfolio selection $P(t,\cdot)$ is the value function 
\begin{equation}
\begin{split}
	V^P(t,w):={\rm E}\Big(-\gamma \widetilde{W}_T^2+\tau^p \widetilde{W}_T\Big)
	+\lambda\int_t^T h(P(s,\cdot))ds,
\end{split}
\end{equation}
with $\tau^P=1+2\gamma {\rm E}\Big(\widetilde{W}_T\Big)$.
Lemma \ref{le:policy_evaluation} shows the explicit expression of the value function $V^P(t,w)$ after a exploratory portfolio selection $P(t,\cdot)$ is given.

\begin{lemma}[Policy evaluation]\label{le:policy_evaluation}
Let $P(t,\cdot)=\mathcal{N}\Big((a_0-w){\bm a_1},e^{a_2}{\bm A_3}\Big)$, $\forall t\in[0,T]$ be an arbitrarily given probability density function,
where $a_0\in\mathbb{R}$, ${\bm a_1}\in\mathbb{R}^n$, $a_2\in\mathbb{R}$, ${\bm A_3}\in\mathbb{R}^{n\times n}$ are time-dependent.
The terminal wealth $\widetilde{W}_T$ 
is defined in \eqref{equ:wealth_RL}
under the exploratory portfolio selection 
$\{P(t,\cdot)\}_{0 \leqslant t \leqslant T}$
with the initial wealth $w^o$.
We have
\begin{itemize}
	\item[(i)] 
${\rm E}\Big(\widetilde{W}_T\Big)$
can be expressed as
\begin{equation}\label{equ:tau_update}
{\rm E}\Big(\widetilde{W}_T\Big) = 
	e^{\int_0^T -{\bm a_1}^\top (\mu-r)ds}
	\Big(\int_0^T a_0{\bm a_1}^\top(\mu-r)e^{\int_0^s {\bm a_1}^\top (\mu-r)dk }ds+w^o \Big).
\end{equation}
	\item[(ii)]
the value function $V^P(t,w)$ can be presented in the form of a quadratic polynomial regarding $w$,
\begin{align*}
	V^P(t,w)
	=-I^{P}(t)\left(w - \frac{H^{P}(t)}{2I^{P}(t)}\right)^{2}
	+\frac{(H^{P}(t))^{2}}{4I^{P}(t)}+G^{P}(t),
\end{align*}
where
\begin{align*}
	&I^P(t)=\gamma e^{\int_{t}^T-\left(2{\bm a_1}^{\top}(\mu - r)-{\bm a_1}^{\top}\Sigma{\bm a_1}\right)ds},\\
	&H^P(t)=e^{\int_{t}^T-{\bm a_1}^{\top}(\mu - r)ds}\left[\tau^P
	+2\gamma\int_{t}^T a_0\left({\bm a_1}^{\top}(\mu - r)-{\bm a_1}^{\top}\Sigma{\bm a_1}\right)e^{\int_{s}^T-\left({\bm a_1}^{\top}(\mu - r)-{\bm a_1}^{\top}\Sigma{\bm a_1}\right)du}ds\right],
\\
	&G^P(t)=\int_t^T
	\Big[H^P(s)a_0{\bm a_1}^\top(\mu-r)-I^P(s)a_0^2{\bm a_1}^\top\Sigma{\bm a_1}\\
	&\hspace{+2cm}
	+\dfrac{\lambda n}{2}\ln{(2\pi e)}+\dfrac{\lambda n}{2}a_2+\dfrac{\lambda}{2}\ln{|{\bm A_3}|} -I^P(s) e^{a_2}{\rm tr}(\Sigma {\bm A_3})
	\Big] ds.
\end{align*}
\end{itemize}
\end{lemma}
{\comm
\begin{proof}
	See Appendix \ref{app:policy_evaluation}.
\end{proof}
}

When it comes to the policy improvement, another exploratory portfolio selection $\widetilde{P}(t,\cdot)$, constructed by $V^P(t,w)$, is introduced to enhance the value function.
$\widetilde{P}(t,\cdot)$ facilitates the improvement of the original exploratory portfolio selection $P(t,\cdot)$ under the given financial market,
and the formulation of $\widetilde{P}(t,\cdot)$ is shown in Lemma \ref{le:policy_improvement}.

\begin{lemma}[Policy improvement]
	\label{le:policy_improvement}
Let $P(t,\cdot)$, $\forall t\in[0,T]$, be an arbitrarily given exploratory portfolio selection, and $V^P(t,w)$ be its value function.
We define another exploratory portfolio selection
\begin{align}\label{equ:tilde_P}
	\widetilde{P}(t,\cdot)
	=\mathcal{N}\Big(\dfrac{\frac{\partial V^{P}}{\partial w}(t,w)}{-\frac{\partial^2 V^{P}}{\partial w^2}(t,w)}\Sigma^{-1}(\mu-r),
	\dfrac{\lambda}{-\frac{\partial^2 V^{P}}{\partial w^2}(t,w)}\Sigma^{-1}\Big),
\end{align}
whose value function is $V^{\widetilde{P}}(t,w)$ with terminal condition $V^{\widetilde{P}}(T,w)=V^P(T,w)$.
Then, we have
\begin{align*}
	V^{\widetilde{P}}(t,w)\geq V^{P}(t,w),\quad \forall (t,w)\in[0,T]\times \mathbb{R}.
\end{align*}
\end{lemma}

\begin{proof}
	See Appendix \ref{app:policy_improvement}.
\end{proof}

Lemma \ref{le:policy_evaluation} and Lemma \ref{le:policy_improvement} suggest that, when choosing an initial exploratory portfolio selection within Gaussian distribution, there are always policies in the Gaussian family capable of completing the policy iteration procedure.
We denote the initial exploratory portfolio selection by $P_0(t,\dot)$, while $V^{P_0}(t,w)$ is the initial value function obtained by the policy evaluation in Lemma \ref{le:policy_evaluation}.
According to Lemma \ref{le:policy_improvement}, the exploratory portfolio selection in the first iteration is updated into $P_1(t,\dot)$.
Proceeding in a step-by-step iterative manner, the sequence of exploratory portfolio selection $\{P_m(t,\cdot)\}$ and the corresponding value function $\{V^{P_m}(t,w)\}$ are obtained, for $m=1,2,\dots$.
It turns out in Theorem \ref{th:sac} that, the sequence of $\{P_m(t,\cdot)\}$ will converge to the optimal exploratory portfolio selection \eqref{equ:theta_n} when $m\rightarrow\infty$.

\begin{theorem}\label{th:sac}
Let $P_0(t,\cdot)=\mathcal{N}\Big((a_0-w){\bm a_1},e^{a_2}{\bm A_3}\Big)$, $\forall t\in[0,T]$, be an arbitrarily given initial exploratory portfolio selection,
where $a_0\in\mathbb{R}$, ${\bm a_1}\in\mathbb{R}^n$, $a_2\in\mathbb{R}$,
${\bm A_3}\in\mathbb{R}^{n\times n}$ are time-dependent.
Then, for the sequence of $\{P_m(t,\cdot)\}$ and the value function $\{V^{P_m}(t,\cdot)\}$, $m=0,1,2,\dots$, we have
\begin{align*}
	&\lim_{m\to\infty}P_m(t,\theta)=P^*(t,\theta),\\
	&\lim_{m\to\infty}V^{P_m}(t,w)=V^*(t,w),
\end{align*}
where $P^*(t,\theta)$ and $V^*(t,w)$ are defined in Theorem \ref{th:v_theta_n}.
\end{theorem}
\begin{proof}
	See Appendix \ref{app:sac}.
\end{proof}

\section{SAC Algorithm Design}\label{se:sac_alg}

The previous discussion about the policy iteration procedure provides clear theoretical guidance for learning the optimal multi-asset MV portfolio selection.
In this section, we develop a reinforcement learning (RL) algorithm, the online Soft Actor-Critic (SAC) algorithm, to translate the theoretical guidance into practical solutions.

According to Lemma \ref{th:theta_cl}, the optimal multi-asset MV portfolio selection 
is intrinsically linked to the parameters 
$\mu-r\in\mathbb{R}^{n\times 1}$, $\Sigma^{-1}\in\mathbb{R}^{n\times n}$, and $K(0,T)\in\mathbb{R}$.
Learning the optimal portfolio selection boils down to learning these three parts of parameters.
However, in the context of multi-asset portfolio selection problem,
learning all parameters in these three parts simultaneously faces challenge.
It leads to numerical instability, thus undermining the reliability and effectiveness of the portfolio.
As a result, we decouple the learning processes.
Specifically,
$\mu-r$ and $\Sigma^{-1}$, both time-dependent, are the immediate factors for the optimal portfolio selection, while $K(0, T)$, a constant over the whole planning investment horizon, serves as a long-term factor. 
Thus, in Section \ref{se:mu_r}, we first develop an algorithm presented in Algorithm \ref{alg:mu_r}, in which the vector $\mu-r$ is learned in each dimension independently.
The inverse covariance matrix $\Sigma^{-1}$ is obtained by the shrinking estimators in \cite{Shi2020Improving}.
Then, in Section \ref{se:alg}, we focus on learning $K(0,T)$ under given the estimation of $\mu-r$ and $\Sigma^{-1}$, which is presented in Algorithm \ref{alg:K}. 
At the end of Section \ref{se:alg}, we combine the learning processes of $\mu-r$, $\Sigma^{-1}$ and $K(0,T)$, and develop the online SAC algorithm, i.e., Algorithm \ref{alg:online}, for continuous-time multi-asset MV portfolio selection.

\subsection{Learning the excess return}\label{se:mu_r}


In $\mu-r\in\mathbb{R}^{n\times 1}$, the excess return
$\mu^{(i)}-r$ is only related to the price data of the $i$-th risky asset, $i=1,\dots,n$.
Thus, in this section, we conduct independent learning process for the excess return of each risky asset.
For $\mu^{(i)}-r$,
we develop an one-dimensional algorithm to learn it, based on the policy evaluation and policy improvement in Section \ref{se:alg_convergence} with an special case of $n=1$.

In the common practice within the field of RL algorithm, the (exploratory) portfolio selection is usually parameterized with (deep) neural networks \citep{Coache2024Reinforcement, Duarte2024Machine}.
Thanks to Theorem \ref{th:v_theta_n} and Theorem \ref{th:sac}, 
at time $t$, we can parameterize the one-dimensional exploratory portfolio selection, which only consists of the riskless asset and the $i$-th risky asset, with the explicit expression:
\begin{equation}
\label{equ:para_p_i}
	p(t,\theta;\phi^{(i)})
	=\mathcal{N}\Big((\dfrac{\phi_1^{(i)}}{2\gamma}-w)
		\phi_4^{(i)}
	\phi_3^{(i)},\dfrac{\lambda}{2}\dfrac{e^{\phi_2^{(i)}\cdot(T-t)}}{\gamma}
	\phi_4^{(i)}
\Big),
\end{equation}
where $\phi^{(i)}=\{\phi_1^{(i)},\phi_2^{(i)},\phi_3^{(i)},\phi_4^{(i)}\}$.
Comparing the expression of $p(t,\theta;\phi^{(i)})$ with the optimal exploratory portfolio selection $P^*(t,\cdot)$ in Theorem \ref{th:v_theta_n}, 
we conclude that,
$\phi_1^{(i)}$, $\phi_2^{(i)}$, $\phi_3^{(i)}$ and $\phi_4^{(i)}$ 
are introduced to learn
$K^{(i)}(0,T)$, $K^{(i)}(t,T)$, $\mu^{(i)}-r$ and $(\sigma^{(i)})^2$, i.e.,
\begin{align}\label{equ:phi_ki}
	\phi_1^{(i)} = e^{K^{(i)}(0,T)\cdot T}+2\gamma w^o,\quad
	\phi_2^{(i)} = K^{(i)}(t,T),\quad
	\phi_3^{(i)} = \mu^{(i)}-r,\quad
	\phi_4^{(i)} = \frac{1}{{(\sigma^{(i)})^2}}.
\end{align}
respectively. And, $\phi_3^{(i)}$ is what we need to obtain.



As known in Section \ref{se:alg_convergence}, the learning process consists of the iterative procedures of policy evaluation and policy improvement.
We start from some initialized values for $\phi^{(i)}$ and then update them iteratively.
For the policy evaluation, given $\phi^{(i)}$, according to Lemma \ref{le:policy_evaluation}, 
the corresponding value function $v^p(t,w)$ is not only related to $p(t,\theta;\phi^{(i)})$ but also to the true value of return rate $\mu^{(i)}(t)$ and volatility $\sigma^{(i)}(t)$ which cannot be obtained directly.
In order to implement the policy evaluation, at time $t$,
according to the form of \eqref{equ:V_n},
we parameterize the corresponding value functions $v^p(t,w)$ as
\begin{equation}
\label{equ:para_v_i}
	v(t,w;\psi^{(i)})
	=-\gamma e^{-\psi_2^{(i)}\cdot(T-t)}(w-\dfrac{\psi_1^{(i)}}{2\gamma})^2+\psi_3^{(i)}+\dfrac{\lambda}{2}\psi_4^{(i)},
\end{equation}
and choose $\psi^{(i)} = \{\psi_1^{(i)},\psi_2^{(i)},\psi_3^{(i)},\psi_4^{(i)}\}$
such that $v(t,w;\psi^{(i)})$ could approximate $v^p(t,w)$ with the available data of the $i$-th risky asset prices.

It is noticed that the value function $v^p(t,w)$
satisfies the dynamic programming
\begin{equation}\label{equ:bellman}
	{\rm E}_{t}\Big(\dfrac{v^p(t+\Delta t,\widetilde{W}_{t+\Delta t})-v^p(t,\widetilde{W}_{t})}{\Delta t}\Big)+\lambda h(p(t,\theta;\phi^{(i)})) = 0
\end{equation}
By collecting $M$ samples for the time-series data of the return rate of the $i$-th risky asset
\begin{equation*}
\{R^{(i,1)},\dots,R^{(i,M)}\},
\end{equation*}
the left-hand side of the dynamic programming \eqref{equ:bellman} can be calculated numerically.
Specifically, for the $k$-th sample, we generate an allocation $\theta_{t}^{(i,k)}$ under the given exploratory portfolio selection $p(t,\theta;\phi^{(i)})$.
The discounted wealth at $t+\Delta t$ time can be simulated by
\begin{align}\label{equ:dW_i}
	W_{t+\Delta t}^{(i,k)}=W_{t}^{(i)}+R^{(i,k)}\theta_{t}^{(i,k)}.
\end{align}
Then, the left-hand side of the dynamic programming \eqref{equ:bellman}
is approximated by
\begin{align*}
	\delta_t:=
	\dfrac{1}{M}\sum_{k=1}^M \dfrac{v(t+\Delta t,W_{t+\Delta t}^{(i,k)};\psi^{(i)})-v(t,W_{t};\bar{\psi}^{(i)})}{\Delta t}
	+\lambda h(p(t,\theta;\phi^{(i)})),
\end{align*}
where $\bar{\psi}^{(i)}$ is the set of parameters in value function \eqref{equ:para_v_i} learned at last time point.
Hence, we define the loss function
\begin{align*}
	L_{t}(\psi^{(i)},\bar{\psi}^{(i)},\phi^{(i)})
	=\dfrac{\Delta t}{2}\delta_{t}^2,
\end{align*}
and update the parameterized value function by
\begin{align}\label{equ:min}
	\psi^{(i)}\leftarrow\arg\min_{\psi^{(i)}} ~L_{t}(\psi^{(i)},\bar{\psi}^{(i)},\phi^{(i)}).
\end{align}

When it comes to the policy improvement, we update the exploratory portfolio selection $p(t,\theta;\phi^{(i)})$ under given updated parameters in $\psi^{(i)}$.
At time $t$, according to Lemma \ref{le:policy_improvement}, the exploratory portfolio selection can be improved into
\begin{align}\label{equ:policy_improvement_i}
	\mathcal{N}\Big((\dfrac{\psi_1^{(i)}}{2\gamma}-w)\dfrac{\mu^{(i)}(t)-r}{(\sigma^{(i)}(t))^2},\dfrac{\lambda}{2}\dfrac{e^{\psi_2^{(i)}(T-t)}}{\gamma}\dfrac{1}{(\sigma^{(i)}(t))^2}\Big).
\end{align}
Comparing \eqref{equ:policy_improvement_i} with 
the parametric form of the exploratory portfolio selection \eqref{equ:para_p_i}, we conduct that the parameters in $p(t,\theta;\phi^{(i)})$ are updated by
\begin{equation}\label{equ:update_phi12_i}
\begin{split}
	\phi_1^{(i)}\leftarrow\psi_1^{(i)},\qquad
	\phi_2^{(i)}\leftarrow\psi_2^{(i)},\qquad
	\phi_4^{(i)}\leftarrow\dfrac{1}{(\widehat{\sigma}^{(i)}(t))^2},
\end{split}
\end{equation}
in which $\widehat{\sigma}^{(i)}(t)$ is obtained by maximum likelihood estimation (MLE) \citep{Campbell1996Econometrics}.
What's more, following \cite{Wang2020Continuous}, 
when given $\phi_1^{(i)}$, $\phi_2^{(i)}$ and $\phi_4^{(i)}$,
the parameter $\phi_3^{(i)}$ can be updated by
\begin{equation}\label{equ:update_phi3_i}
	\phi_3^{(i)}\leftarrow\arg\max_{\phi_3^{(i)}} L_{t}(\psi^{(i)},\bar{\psi}^{(i)},\phi^{(i)}).
\end{equation}
The pseudocode of iterative learning procedure for $\mu^{(i)}-r$ is summarized in Algorithm \ref{alg:mu_r}.
After learned $\mu^{(i)}-r$ for each risky asset,
the excess expected return vector is assembled by
\begin{equation}\label{equ:phi3_combine}
	\mu-r\leftarrow\begin{bmatrix} \phi_3^{(1)},\dots,\phi_3^{(n)} \end{bmatrix}^\top.
\end{equation}

\renewcommand{\algorithmicrequire}{\textbf{Input:}}
\renewcommand{\algorithmicensure}{\textbf{Output:}}
\begin{breakablealgorithm}\label{alg:mu_r}
\caption{The Learning Process of $\mu^{(i)}-r$}
\begin{algorithmic}[1]
\REQUIRE 
The initialized parameters $\phi^{(i)}$ in exploratory portfolio selection;
The initialized parameters $\psi^{(i)}$ in the value function;
The $t$ time discounted wealth $W_{t}$;
The time-series data of return rates of the $i$-th risky asset $\{R^{(i,1)},\dots,R^{(i,M)}\}$.
\ENSURE 
The learned parameter $\mu^{(i)}-r=\phi_3^{(i)}$.
\renewcommand{\algorithmicensure}{\textbf{Procedure:}}
\ENSURE  
\FOR {$k=1:M$}
	\STATE{Sample an allocation $\theta_{t}^{(i,k)}\sim p(t,W_{t},\theta;\phi^{(i)})$.}
	\STATE{Simulate the discounted wealth $W_{t+\Delta t}^{(i,k)}$ using $R^{(i,k)}$ with \eqref{equ:dW_i}.}
\ENDFOR
\STATE{Update the parameters $\psi^{(i)}$ in the value function with \eqref{equ:min}.}
\STATE{Update the parameters $\phi^{(i)}$ in the exploratory portfolio selection with \eqref{equ:update_phi12_i} and \eqref{equ:update_phi3_i}.}
\end{algorithmic}
\end{breakablealgorithm}

\subsection{Learning the average profitability of risky assets}\label{se:alg}

After obtained the estimation of excess expected return vector $\widehat{\mu}-r$ and inverse covariance matrix $\widehat{\Sigma}^{-1}$, in this section, we focus on the learning process of $K(0,T)$ to complete the multi-asset MV optimization framework.
And, a $n$-dimensional algorithm, which operates through the iterative process of policy evaluation and policy improvement presented in Section \ref{se:alg_convergence} with $n$ risky assets, is designed.

At time $t$, we parameterize the multi-asset exploratory portfolio selection with an explicit expression:
\begin{equation}\label{equ:para_p}
	p(t,\theta;\phi)
	=\mathcal{N}\Big((\dfrac{\phi_1}{2\gamma}-w)
		\widehat{\Sigma}^{-1}(\widehat{\mu}-r),
	\dfrac{\lambda}{2}\dfrac{e^{\phi_2\cdot(T-t)}}{\gamma}\widehat{\Sigma}^{-1}\Big),
\end{equation}
{\comm where $\phi=\{\phi_1,\phi_2\}$.}
Comparing the expression of $p(t,\theta;\phi)$ with the optimal multi-asset exploratory portfolio selection $P^*(t,\cdot)$ in Theorem \ref{th:v_theta_n}, 
we conclude that, $\phi_1$ is introduced to learn $K(0,T)$, and $\phi_2$ is introduced to learn $K(t,T)$
\begin{align}\label{equ:phi_k}
	\phi_1 = e^{K(0,T)\cdot T}+2\gamma w^o,\quad
	\phi_2 = K(t,T).
\end{align}

For the policy evaluation, similar to the approach in Section \ref{se:mu_r}, we approximate the value function of the multi-asset exploratory portfolio selection \eqref{equ:para_p} with
\begin{equation}\label{equ:para_v}
	v(t,w;\psi)
	=-\gamma e^{-\psi_2\cdot(T-t)}(w-\dfrac{\psi_1}{2\gamma})^2+\psi_3+\dfrac{\lambda}{2}\psi_4,
\end{equation}
where $\psi=\{\psi_1,\psi_2,\psi_3,\psi_4\}$.
As the historical data of risky asset prices can be reused,
we collect $M$ samples for time-series data of the return rate of $n$ risky assets
\begin{align*}
	\{R^1,\dots,R^M\},
\end{align*}
in which $R^k=\begin{bmatrix} R^{(1,k)} &\cdots &R^{(n,k)} \end{bmatrix}^\top\in\mathbb{R}^{n\times 1}, k=1,\dots,M$.
For the $k$-th sample, we generate an allocation $\Theta_{t}^k\in\mathbb{R}^{n\times 1}$ under the given multi-asset exploratory portfolio selection $p(t,\theta;\phi)$, and simulate the discounted wealth at $t+\Delta t$ time with
\begin{equation}\label{equ:dW}
	W_{t+\Delta t}^k=W_{t}+(R^k)^\top\Theta_{t}^{k}.
\end{equation}
By defining
\begin{equation*}
	\delta_{t}(\psi,\bar{\psi},\phi)
	=\dfrac{1}{M}\sum_{k=1}^M \dfrac{v(t+\Delta t,W_{t+\Delta t}^k;\psi)-v(t,W_{t};\bar{\psi})}{\Delta t}
	+\lambda h(p(t,\theta;\phi))
\end{equation*}
and the loss function
\begin{equation*}
	L_{t}(\psi,\bar{\psi},\phi)
	=\dfrac{\Delta t}{2}
	\delta_t^2,
\end{equation*}
in which $\bar{\psi}$ is the set of parameters in value function \eqref{equ:para_v} learned at last time point,
the parameterized value function can be updated by
\begin{align}\label{equ:multi_min}
	\psi\leftarrow\arg\min_{\psi} ~L_{t}(\psi,\bar{\psi},\phi).
\end{align}

For the policy improvement in the $n$-dimensional algorithm, under given updated parameters in $\psi$, according to Lemma \ref{le:policy_improvement}, the multi-asset exploratory portfolio selection can be improved into
\begin{align}\label{equ:policy_improvement}
	\mathcal{N}\Big((\dfrac{\psi_1}{2\gamma}-w)
		\widehat{\Sigma}^{-1}(\widehat{\mu}-r),
	\dfrac{\lambda}{2}\dfrac{e^{\psi_2(T-t)}}{\gamma}\widehat{\Sigma}^{-1}\Big).
\end{align}
Comparing \eqref{equ:policy_improvement} with 
the parametric form in \eqref{equ:para_p}, we conduct that the parameters in multi-asset exploratory portfolio selection $p(t,\theta;\phi)$ are updated by
\begin{align}\label{equ:update_phi}
	\phi_1\leftarrow\psi_1,
	\qquad
	\phi_2\leftarrow\psi_2.
\end{align}
Thus, at time $t$, the pseudocode of iterative learning procedure for $K(0,T)$ can be summarized in Algorithm \ref{alg:K}.
\renewcommand{\algorithmicrequire}{\textbf{Input:}}
\renewcommand{\algorithmicensure}{\textbf{Output:}}
\begin{breakablealgorithm}\label{alg:K}
\caption{The Learning Process of $K(0,T)$}
\begin{algorithmic}[1]
\REQUIRE 
The values of $\widehat{\mu}-r$ and $\widehat{\Sigma}^{-1}$;
The initialized parameters $\phi_1,\phi_2$;
The initialized parameters $\psi_1,\psi_2,\psi_3,\psi_4$;
The $t$ time discounted wealth $W_{t}$;
The time-series data of returns rates $\{R^1,\cdots,R^M\}$.
\ENSURE 
The learned parameter $K(0,T) = \frac{1}{T}\ln( \phi_1-2\gamma w^o)$.
\renewcommand{\algorithmicensure}{\textbf{Procedure:}}
\ENSURE  
\FOR {$k=1:M$}
	\STATE{Sample an allocation $\Theta_{t}^k\sim p(t,W_{t},\theta;\phi)$.}
	\STATE{Simulate the next time discounted wealth $W_{t+\Delta t}^k$ using $R^k$ with \eqref{equ:dW}.}
\ENDFOR
\STATE{Update the parameters $\psi$ in the value function with \eqref{equ:multi_min}.}
\STATE{Update the parameters $\phi$ in the exploratory portfolio selection with \eqref{equ:update_phi}.}
\end{algorithmic}
\end{breakablealgorithm}

Finally, we can develop the online SAC algorithm, Algorithm \ref{alg:online}, for learning the continuous-time multi-asset MV portfolio selection in a discrete-time setting. 
We divide the investment horizon $[0,T]$ into $N$ time intervals $[t_j,t_{j+1})$, $j = 0, 1, \dots, N-1$, where $t_0 = 0$ and $t_N = T$.
At each time point, the portfolio selection is implemented with the currently learned parameters, and the wealth at the next time point is obtained.
We reiterate that the exploratory portfolio selections are used for learning, and the mean of the learned multi-asset exploratory portfolio selection is used when implementing.

In Algorithm \ref{alg:online}, parameters are updated every $m$ time points.
When performing the updates,
we use the values obtained from the previous update as the initial values for the current update.
For the $i$-th risky asset, Algorithm \ref{alg:mu_r} is called to learn $\mu^{(i)}-r$.
Subsequently, $\widehat{\mu}-r$ is obtained by \eqref{equ:phi3_combine},
and $\widehat{\Sigma}^{-1}$ is obtained by the shrinking technique in \cite{Shi2020Improving}. 
Thereafter, $\widehat{\mu}-r$ and $\widehat{\Sigma}^{-1}$ are then used as inputs for Algorithm \ref{alg:K},
and all the parameters in optimal multi-asset MV portfolio selection \eqref{equ:theta_n_cl} can be learned.

\renewcommand{\algorithmicrequire}{\textbf{Input:}}
\renewcommand{\algorithmicensure}{\textbf{Output:}}
\begin{breakablealgorithm}\label{alg:online}
\caption{The Optimal multi-asset MV Portfolio Selection with Online SAC Algorithm}
\begin{algorithmic}[1]
\REQUIRE 
Investment horizon $T$;
Time intervals $[t_j,t_{j+1})$, $j = 0, 1, \dots, N-1$;
Initial Wealth $w^{o}$. 
\ENSURE 
The optimal multi-asset MV portfolio selection process $\{\Theta_{t_j}^*\}_{j=0}^{N-1}$;
The corresponding wealth process $\{W_{t_j}\}_{j=0}^N$.
\renewcommand{\algorithmicensure}{\textbf{Procedure:}}
\ENSURE 
\STATE{Set the learning cycle $m$.}
\FOR {$j=0:(N-1)$}
\IF {$j \equiv 0 \pmod{m}$}
\FOR {$i=1:n$}
	\STATE{Update $\phi_3^{(i)}$ by Algorithm \ref{alg:mu_r}.}
\ENDFOR
	\STATE{Set $\widehat{\mu}-r\leftarrow\begin{bmatrix} \phi_3^{(1)},\dots,\phi_3^{(n)} \end{bmatrix}^\top$.}
	\STATE{Estimate $\widehat{\Sigma}^{-1}$ by the shrinking technique in \cite{Shi2020Improving}.}
	\STATE{Update $\phi_1$ by Algorithm \ref{alg:K} with $\widehat{\mu}-r$ and $\widehat{\Sigma}^{-1}$.}
\ENDIF
\STATE{Implement the optimal multi-asset MV portfolio selection at $t_j$ time by
\vspace{-.5em}
\begin{equation*}
	\Theta_{t_j}^*=(\dfrac{\phi_1}{2\gamma}-W_{t_j})\widehat{\Sigma}^{-1}(\widehat{\mu}-r).
\end{equation*}
\vspace{-1em}}
\STATE{Observe the discounted wealth $W_{t_{j+1}}$ at time $t_{j+1}$ from the financial market.}
\ENDFOR
\end{algorithmic}
\end{breakablealgorithm}

\section{Numerical Study}

In this section, we conduct numerical experiments under various simulated and real financial markets to demonstrate the superiority of our SAC algorithm.
The risk aversion coefficient is taken as $\gamma=1.5$ \citep{Kydland1982Time}.
The exploration weight $\lambda$ is exogenous and pre-specified by the SAC agent.
Here, we set $\lambda=1$ and refer the interested readers to \cite{Dai2023Learning} for a detailed description of the value of $\lambda$.

\subsection{The stationary market case}\label{se:stationary}

A key advantage of a simulation study is that we have the ground truth (``omniscient'') values to compare against the learning results. 
In the stationary market case, we investigate the convergence of 
the estimation of $\mu-r$ and $K(0,T)$ given by Algorithm \ref{alg:mu_r} and Algorithm \ref{alg:K}.

The sample paths of risky assets prices are generated from geometric brownian motion \eqref{equ:S} with
\begin{align*}
\mu-r=
\begin{bmatrix}
	0.06 &0.08
\end{bmatrix}^\top,
\qquad
\sigma^{(1)}=0.1,
\qquad
\sigma^{(2)}=0.15,
\qquad
\rho^{(12)}=\rho^{(21)}=0.1,
\end{align*}
which are usually considered as “typical” stocks for simulation \citep{Hutchinson1994nonparametric}.
We generate a training dataset with daily data for 2,500 months.
{\comm First, the parameters $\widehat{\mu}-r$ and $\widehat{\Sigma}^{-1}$ are obtained by Maximum Likelihood Estimation (MLE) according to the whole training dataset, while the parameters in $\phi^{(i)}$ and $\phi$ are initialized by \eqref{equ:phi_ki} and \eqref{equ:phi_k} with $\widehat{\mu}-r$, $\widehat{\Sigma}^{-1}$.}
Then, at each learning episode, we randomly sample a consecutive one-month subsequence from the training dataset, and $\mu-r$ and $K(0,T)$ are learned by Algorithm \ref{alg:mu_r} and Algorithm \ref{alg:K}, respectively.

Figure \ref{fig:convergence} illustrates the convergence of the relative errors for $\mu^{(1)}-r$, $\mu^{(2)}-r$ and $K(0,T)$. 
In fact, Algorithm \ref{alg:mu_r} and Algorithm \ref{alg:K}, initialized by MLE, demonstrates significant improvements in parameter estimation accuracy. 
Specifically, after 3,000 learning episodes, the relative errors for \(\mu^{(1)} - r\) and \(\mu^{(2)} - r\) are reduced to around 1\% and 3\%, respectively, and to around 4\% for $K(0,T)$.
Notably, during the learning episodes, the relative errors of all parameters show a steady decreasing trend. 
This stable convergence pattern emphasizes the effectiveness of the proposed SAC algorithm and has the potential to improve the out-of-sample performance for the multi-asset MV portfolio selection.

\begin{figure}[H]
\centering
\subfigure[]{
	\includegraphics [width=0.48\textwidth] {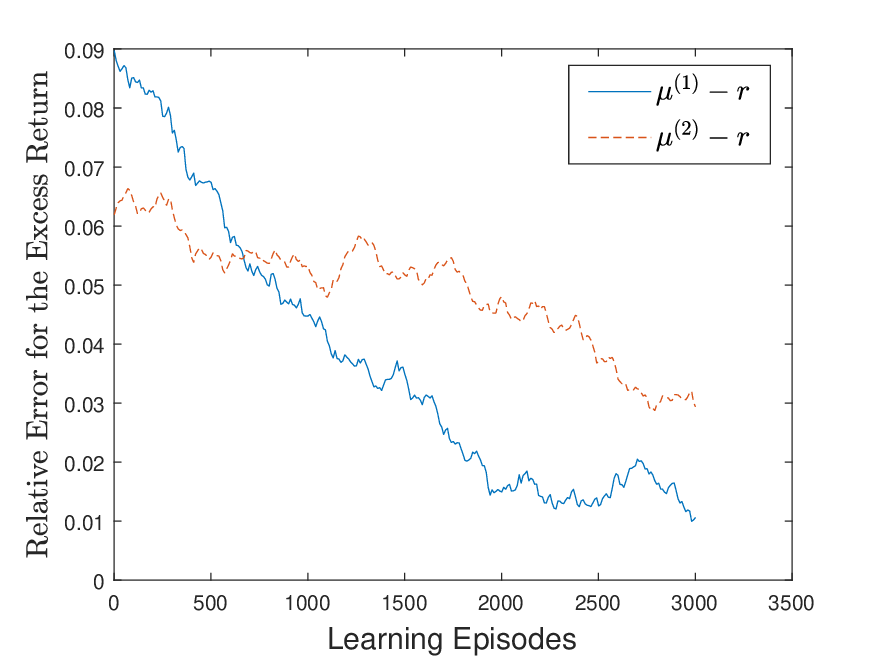}
}
\subfigure[]{
	\includegraphics [width=0.48\textwidth] {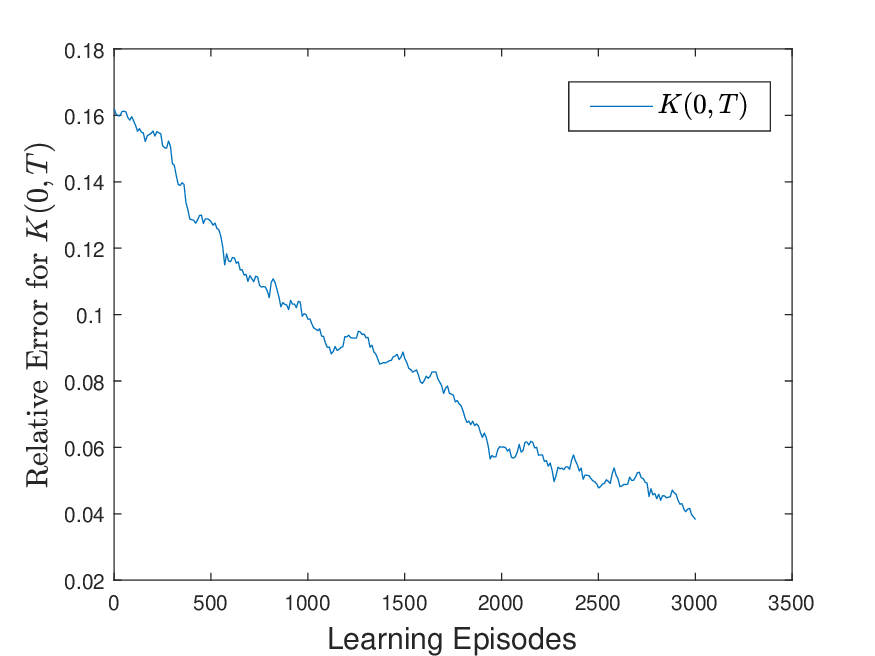}
}
\caption{The relative errors.}\label{fig:convergence}
\end{figure}

Next, we show the robustness of the convergence of Algorithm \ref{alg:mu_r} and Algorithm \ref{alg:K}.
Since the correlation coefficients between risky assets are crucial factors differentiating multi-asset financial markets from single-asset ones,
we carry out experiments with various correlation coefficients $\rho^{(12)}$ (or $\rho^{(21)}$) between the two risky assets.
In Figure \ref{fig:rho}, we report the relative errors of $\mu^{(1)}-r$, $\mu^{(2)}-r$ and $K(0,T)$ as the learning episodes increases.
It is shown that, in all the simulated financial markets, the relative errors of $\mu^{(1)}-r$, $\mu^{(2)}-r$ and $K(0,T)$ decrease in a consistent and stable manner.
This convergence pattern indicates the reliability and adaptability of Algorithm \ref{alg:mu_r} and Algorithm \ref{alg:K} in different market conditions.
\begin{figure}[H]
\centering
\subfigure[$\rho^{(12)}=\rho^{(21)}=0$]
{
	\begin{minipage}[t]{0.31\linewidth}
		\centering
		\includegraphics [width=0.98\textwidth] {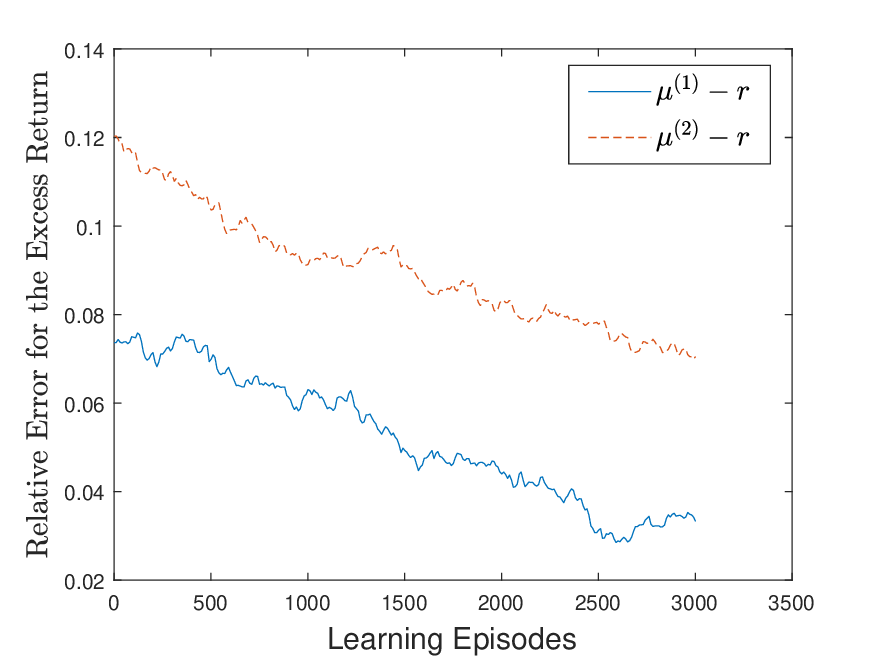}\\
		\vspace{0.02cm}
		\includegraphics [width=0.98\textwidth] {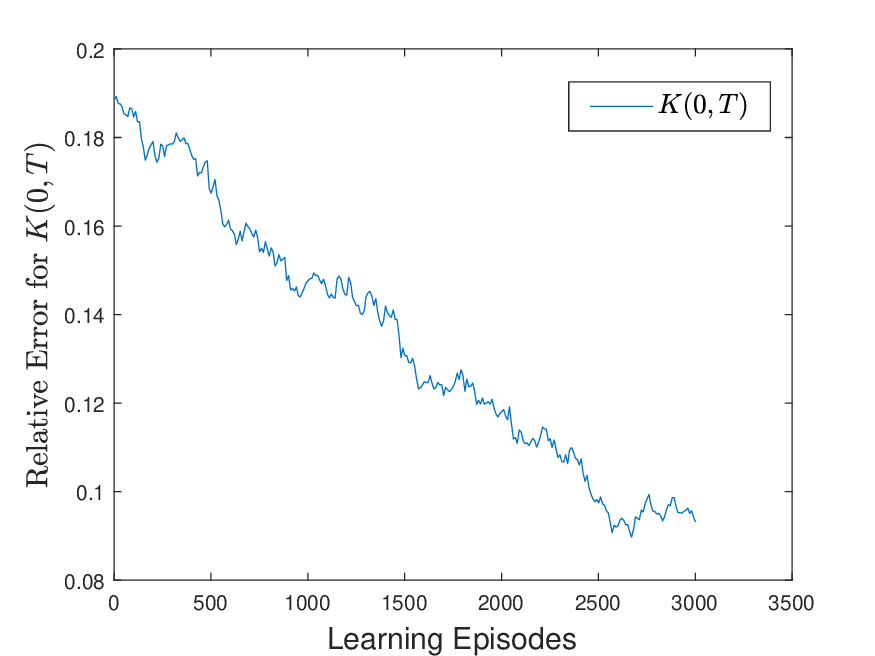}\\
		\vspace{0.02cm}
	\end{minipage}
}
\subfigure[$\rho^{(12)}=\rho^{(21)}=0.05$]
{
	\begin{minipage}[t]{0.31\linewidth}
		\centering
		\includegraphics [width=0.98\textwidth] {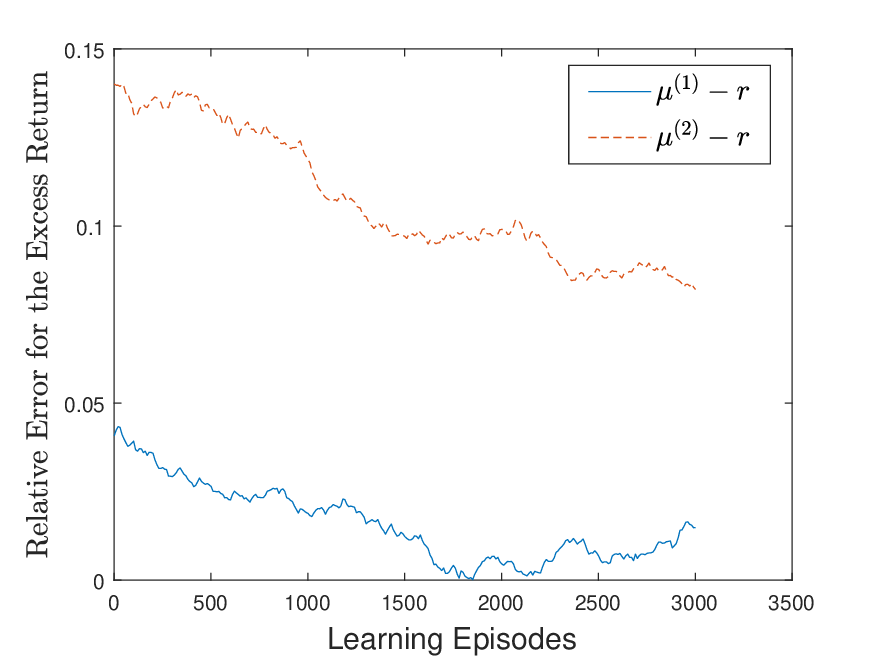}\\
		\vspace{0.02cm}
		\includegraphics [width=0.98\textwidth] {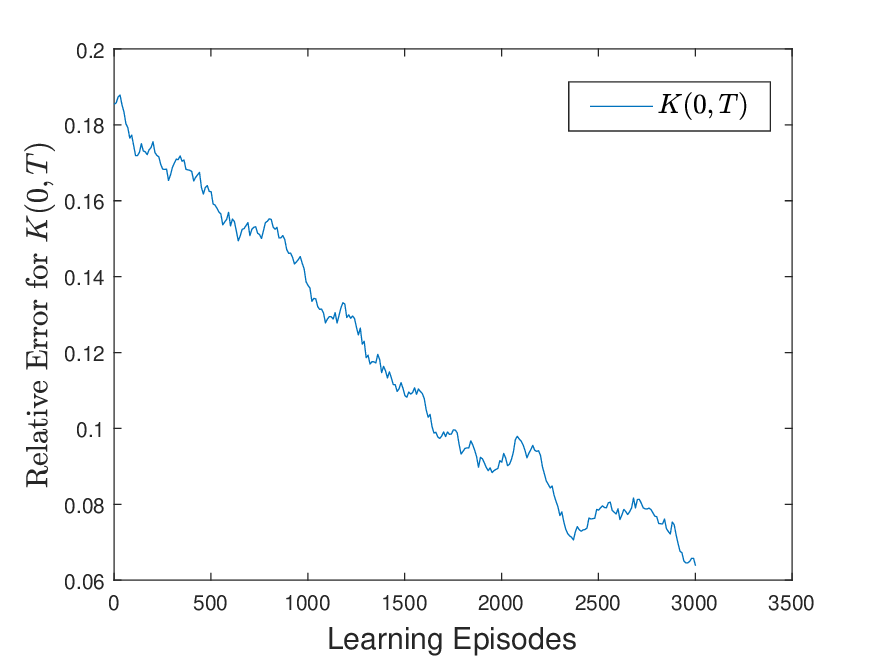}\\
		\vspace{0.02cm}
	\end{minipage}
}
\subfigure[$\rho^{(12)}=\rho^{(21)}=0.15$]
{
	\begin{minipage}[t]{0.31\linewidth}
		\centering
		\includegraphics [width=0.98\textwidth] {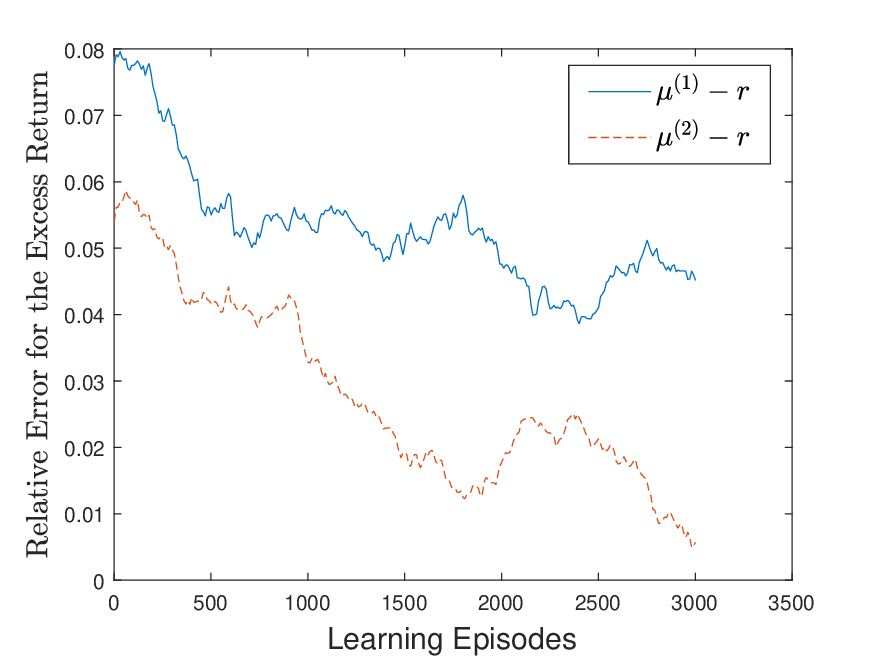}\\
		\vspace{0.02cm}
		\includegraphics [width=0.98\textwidth] {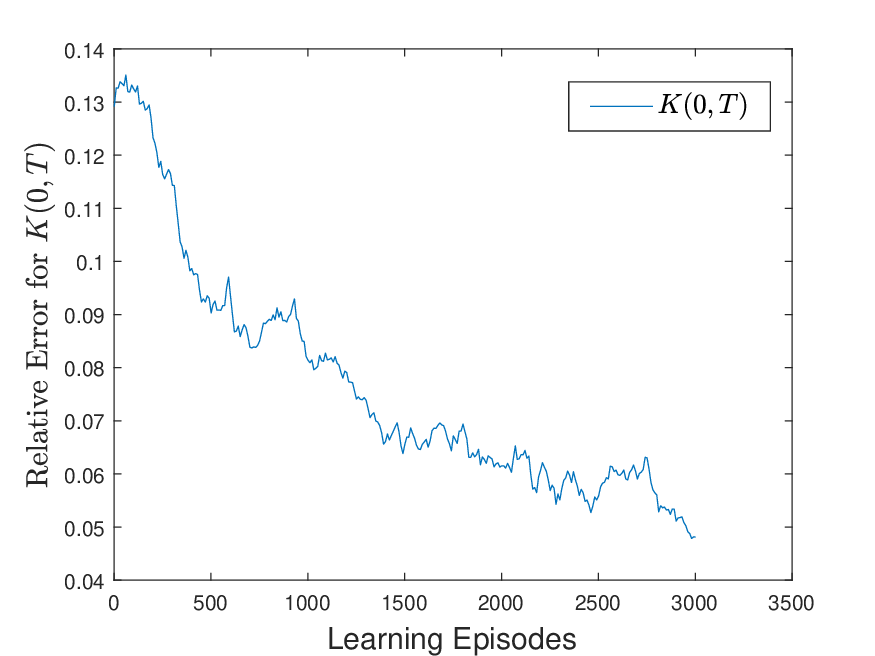}\\
		\vspace{0.02cm}
	\end{minipage}
}
\caption{The relative errors under various simulated financial markets}\label{fig:rho}
\end{figure}

Finally, we show the stability of Algorithm \ref{alg:K}.
According to Theorem \ref{th:K_k}, $K(0,T)$ can also be derived using a ``Combination'' method, 
in which $K^{(i)}(0,T)$ is learned by Algorithm \ref{alg:mu_r} and $K(0,T)$ is combined through \eqref{equ:K_k}.
In contrast, Algorithm \ref{alg:K} learns $K(0,T)$ as a whole in the multi-asset financial market.
In Figure \ref{fig:K}, we compare the performance of these two methods.
In subfigure (a), 
the relative error of $K(0, T)$ obtained by Algorithm \ref{alg:K} continuously and steadily decreases as the number of learning episodes increases.
Conversely, the relative error of $K(0,T)$ obtained by ``Combination'' method is neither stable nor convergent.

The relationship between the relative error of $K(0,T)$ and the relative error of $\Sigma^{-1}$ for both methods is depicted in subfigure (b).
Let's define the relative error of $\Sigma^{-1}$ as $\frac{\|\hat{\Sigma}^{-1}-\Sigma^{-1}\|}{\|\Sigma^{-1}\|}$, where $\hat{\Sigma}^{-1}$ is the estimated value of $\Sigma^{-1}$ and $\|\cdot\|$ represents the 2-norm of a matrix.
In subfigure (b), it can be observed that, for Algorithm \ref{alg:K}, there exists a relatively weak correlation between the relative error of $K(0,T)$ and the relative error of $\Sigma^{-1}$. 
Specifically, despite significant fluctuations in the relative error of $\Sigma^{-1}$ within a given range, the relative error of $K(0, T)$ maintains remarkable stability, further demonstrate the potential of Algorithm \ref{alg:K} in real-world applications.

\begin{figure}[H]
\centering
\subfigure[]{
	\includegraphics [width=0.48\textwidth] {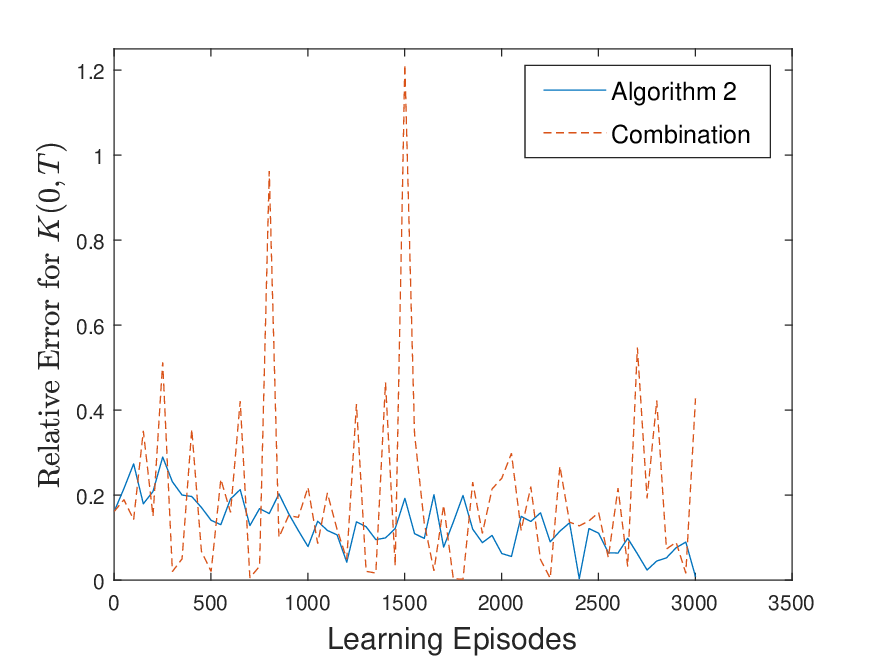}
}
\subfigure[]{
	\includegraphics [width=0.48\textwidth] {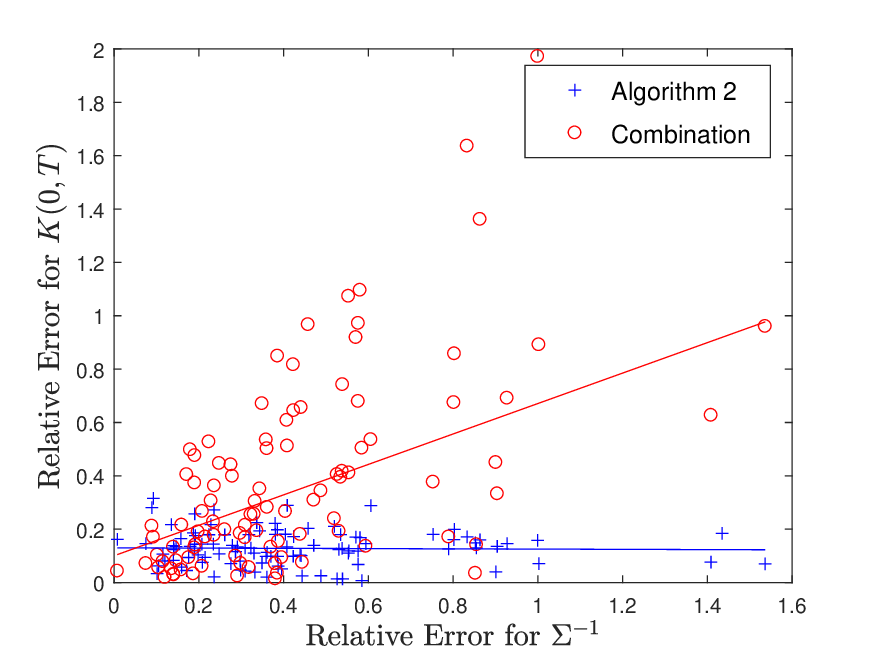}
}
\caption{The estimation of $K(0,T)$ 
using Algorithm \ref{alg:K} and the ``Combination'' method.}\label{fig:K}
\end{figure}

\subsection{The real financial market case}

In the real financial market case, we study the dynamic allocation among a riskless asset and multiple risky assets.
We consider the risk-free interest rate $r=0.02$ and the initial wealth $w^o = 1$.
The planning investment horizon is set to be $T=\frac{21}{252}$ year (one month), 
and rebalancing of multi-asset MV portfolio selection takes place every day ($N = 21$) with transaction cost $c=3‰$ \citep{Balduzzi1999Transaction}.
In Algorithm \ref{alg:online}, the parameters $\psi$ and $\phi$ are updated every $m=\frac{5}{252}$ year (one week). 
The learned values are then used throughout the next week.
We allow leverage and borrowing, and truncate the proportion $\frac{\sum_{i=1}^n|\theta_t^{(i)}|}{W_t}$ to be in the interval $[-1, 2]$, $\forall t\in[0,T]$.

We compare the portfolio selection based on the Algorithm \ref{alg:online}, denoted by ``SAC'', with the broad-market index as well as two other portfolio selections:
\begin{description}[leftmargin = 3.5em, labelwidth = 1em]
\item[Plug-in]
This portfolio selection is obtained by traditional paradigm.
It follows a rolling time window to form the MLE for the model parameters, and then substitute the resulting MLE into the analytical solutions \eqref{equ:theta_n_cl} for the portfolios. 
\item[B-H\quad\,]
The naive ``buy-and-hold'' portfolio selection, which equally invests wealth into $n$ risky assets at each rebalancing date.
This portfolio selection does not involve any estimation or optimization.
\end{description}
The above four portfolio selections are computed across different real financial markets, which are widely used, as listed in Table \ref{table:data}. 
For each real financial market, we take the data of risky asset prices from 2000-01-01 to 2024-12-30, 300 months in total, and use the first 144 months (12 years) for training and leave the rest 156 months (13 years) for testing.
\setlength\LTcapwidth{\textwidth}
\begin{longtable}{p{2.5cm}|p{7cm}|p{1.25cm}|p{2.5cm}}
\caption{Data Description}\label{table:data}\\ 
\toprule
Abbreviation  &Description &$n$ &The broad-market index
  \\\hline
  29DJI
  &The components of DJI which are listed before 2000-01-01
  &29
  &DJI
  \\\hline        
  57NASDAQ
  &The components of NASDAQ100 which are listed before 2000-01-01
  &57
  &NASDAQ100
  \\\hline        
  340SP
  &The components of S\&P500 which are listed before 2000-01-01
  &340
  &S\&P500
  \\
\bottomrule
\end{longtable}

\noindent
The testing performance of portfolio selections is assessed based on the
following criteria:
\begin{itemize}
\item
monthly Mean of investment return rate (MEAN)
\item
monthly Standard Deviation of investment return rate (STD)
\item
annualized Certainty-Equivalent Return (CEQ) \citep{DeMiguel2007Optimal}

\item
annualized Sharpe Ratio (SR) \citep{Sharpe1994The}

\item
daily Turnover Rate (TR) \citep{Kirby2012Its}

\item
annualized Certainty-Equivalent Return adjusted 
under the Transaction Costs 

\item
annualized Sharpe Ratio adjusted 
under the Transaction Costs 
\end{itemize}
Among these criteria, MEAN measures the investment return of the portfolio, while STD measures the investment risk.
CEQ represents the guaranteed return an investor would accept rather than adopting the portfolio, theoretically linked to mean-variance utility \eqref{model_cl} under unit initial wealth ($W_0=w^o=1$). 
SR normalizes excess returns by volatility, providing a risk-adjusted performance for the portfolio.  
TR reflects portfolio stability, with lower values indicating reduced transaction costs.  
CEQ\_TR and SR\_TR extend the measures of CEQ and SR by explicitly incorporating transaction costs to align with real-world implementation.

The results are reported in Table \ref{table:DJI}-\ref{table:SP500_340}.
In these tables, it is evident that the ``SAC'' portfolio selection always yields the highest average investment return rate, significantly outperforming the other three portfolios.
Additionally, it also attains the highest annualized CEQ and SR, followed by the ``B-H'' portfolio selection and the corresponding broad-market index.
In contrast, the ``Plug-in'' approach performs the worst in various criteria, not only in average terminal wealth but also in annualized CEQ and SR.
When considering the transaction costs, the superiority of the ``SAC'' portfolio selection is clear.
It consistently outperforms all the other portfolio selections by large margins in annualized CEQ\_TR and SR\_TR.

\setlength\LTcapwidth{\textwidth}
\begin{longtable}{p{1.75cm}|p{2.25cm}|p{2.25cm}|p{2.25cm}|p{2.25cm}}
\caption{Comparison of different portfolio selections in the real financial market of 29DJI}\label{table:DJI}\\ 
\toprule
&\thead{SAC} &\thead{Plug-in} &\thead{B-H} &\thead{DJI}
  \\\midrule
  \thead{MEAN}
  &\thead{$0.0318$}
  &\thead{$0.0067$\\$(p=0.0010)$}
  &\thead{$0.0135$\\$(p=0.0069)$}
  &\thead{$0.0086$\\$(p=0.0010)$}
  \\\hline        
  \thead{STD}
  &\thead{$0.0756$}
  &\thead{$0.0560$}
  &\thead{$0.0356$}
  &\thead{$0.0425$}
  \\\midrule
  \thead{CEQ}
  &\thead{$0.2789$}
  &\thead{$0.0242$}
  &\thead{$0.1394$}
  &\thead{$0.0709$}
  \\\hline
  \thead{SR}
  &\thead{$1.3810$}
  &\thead{$0.3131$}
  &\thead{$1.1535$}
  &\thead{$0.5665$}
  \\\midrule
  \thead{TR}
  &\thead{$0.0939$}
  &\thead{$0.2478$}
  &\thead{$0.0076$}
  &\thead{$0.0000$}
  \\\hline
  \thead{CEQ\_TR}
  &\thead{$0.2072$}
  &\thead{$-0.1942$}
  &\thead{$0.1336$}
  &\thead{$0.0709$}
  \\\hline
  \thead{SR\_TR}
  &\thead{$1.1097$}
  &\thead{$-0.8119$}
  &\thead{$1.1062$}
  &\thead{$0.5665$}
  \\
\bottomrule
\end{longtable}

\setlength\LTcapwidth{\textwidth}
\begin{longtable}{p{1.75cm}|p{2.25cm}|p{2.25cm}|p{2.25cm}|p{2.25cm}}
\caption{Comparison of different portfolio selections in the real financial market of 57NASDAQ}\label{table:NASDAQ}\\ 
\toprule
&\thead{SAC} &\thead{Plug-in} &\thead{B-H} &\thead{NASDAQ100}
  \\\midrule
  \thead{MEAN}
  &\thead{$0.0334$}
  &\thead{$-0.0027$\\$(p<0.0001)$}
  &\thead{$0.0154$\\$(p=0.0260)$}
  &\thead{$0.0127$\\$(p=0.0123)$}
  \\\hline        
  \thead{STD}
  &\thead{$0.0860$}
  &\thead{$0.0569$}
  &\thead{$0.0465$}
  &\thead{$0.0503$}
  \\\midrule
  \thead{CEQ}
  &\thead{$0.2678$}
  &\thead{$-0.0916$}
  &\thead{$0.1461$}
  &\thead{$0.1078$}
  \\\hline
  \thead{SR}
  &\thead{$1.2787$}
  &\thead{$-0.2697$}
  &\thead{$1.0242$}
  &\thead{$0.7652$}
  \\\midrule
  \thead{TR}
  &\thead{$0.0918$}
  &\thead{$0.2562$}
  &\thead{$0.0092$}
  &\thead{$0.0000$}
  \\\hline
  \thead{CEQ\_TR}
  &\thead{$0.1974$}
  &\thead{$-0.3609$}
  &\thead{$0.1391$}
  &\thead{$0.1078$}
  \\\hline
  \thead{SR\_TR}
  &\thead{$1.0433$}
  &\thead{$-1.4125$}
  &\thead{$0.9811$}
  &\thead{$0.7652$}
  \\
\bottomrule
\end{longtable}

\setlength\LTcapwidth{\textwidth}
\begin{longtable}{p{1.75cm}|p{2.25cm}|p{2.25cm}|p{2.25cm}|p{2.25cm}}
\caption{Comparison of different portfolio selections in the real financial market of 340SP}\label{table:SP500_340}\\ 
\toprule
&\thead{SAC} &\thead{Plug-in} &\thead{B-H} &\thead{S\&P500}
  \\\midrule
  \thead{MEAN}
  &\thead{$0.0440$}
  &\thead{$0.0112$\\$(p=0.0009)$}
  &\thead{$0.0129$\\$(p=0.0005)$}
  &\thead{$0.0110$\\$(p=0.0002)$}
  \\\hline        
  \thead{STD}
  &\thead{$0.1033$}
  &\thead{$0.0654$}
  &\thead{$0.0376$}
  &\thead{$0.0355$}
  \\\midrule
  \thead{CEQ}
  &\thead{$0.3451$}
  &\thead{$0.0574$}
  &\thead{$0.1303$}
  &\thead{$0.1094$}
  \\\hline
  \thead{SR}
  &\thead{$1.4425$}
  &\thead{$0.5054$}
  &\thead{$1.0426$}
  &\thead{$0.9121$}
  \\\midrule
  \thead{TR}
  &\thead{$0.1112$}
  &\thead{$0.2748$}
  &\thead{$0.0088$}
  &\thead{$0.0000$}
  \\\hline
  \thead{CEQ\_TR}
  &\thead{$0.2605$}
  &\thead{$-0.1742$}
  &\thead{$0.1236$}
  &\thead{$0.1094$}
  \\\hline
  \thead{SR\_TR}
  &\thead{$1.2091$}
  &\thead{$-0.4784$}
  &\thead{$0.9912$}
  &\thead{$0.9121$}
  \\
\bottomrule
\end{longtable}

\section{Conclusion}

The traditional paradigm for the mean–variance (MV) analysis often predicts model parameters first and then optimizes portfolios.
The performance of the traditional paradigm is poor, especially when the scale of portfolio selection is large.
Following \cite{Wang2020Continuous}, in this paper, we design an online soft actor-critic (SAC) algorithm for the portfolio in multi-asset time-varying financial markets, which can improve the out-of-sample performance of it.
In order to further improve the learning accuracy and increase the stability of the multi-asset SAC algorithm, we separate the model parameters and learn them with decoupled processes.
Numerical studies in the simulated and real financial markets show the superiority of the portfolio using our SAC algorithm.

Possible directions for future work include an combination of the SAC algorithm and Deep Neural Network (DNN), which allows portfolio selection problems without analytic expressions to be dealt with. 
In particular, Tensor Neural Network (TNN) proposed by \cite{Wang2024Computing} can be considered, which demonstrates advantages in handling high-dimensional problems due to its unique architecture.
In this way, a wider range of financial problems, such as those involving nonlinear utility functions and diverse investment constraints, can be effectively addressed.
These questions are left for further investigations.

\vspace{+1cm}
\noindent{\bf\Large{Acknowledgments}}
\vspace{+.5cm}

This project was supported in part by 
the National Basic Research Program (12271395), 
the Humanities and Social Science Research Program of the Ministry of Education of China (22YJAZH156), 
the Innovation Team Project for Ordinary University in Guangdong Province, China (2023WCXTD022),
the Excellent Young Teacher Supporting Program of Tianjin University of Finance and Economics, China,
and National Natural Science Foundation of China (12301610).

\vspace{+1cm}
\noindent{\bf\Large{Declaration of No Competing Interests}}
\vspace{+.5cm}

The authors declare that they have no known competing financial interests or personal relationships that could have appeared to influence the work reported in this paper.

\bibliography{MV_tc-15}

\begin{thebibliography}{}

\bibitem[Aquino et~al., 2023]{Aquino2023Portfolio}
Aquino, L. D.~G., Sornette, D., and Strub, M.~S. (2023).
\newblock Portfolio selection with exploration of new investment assets.
\newblock {\em European Journal of Operational Research}, 310(2):773--792.

\bibitem[Balduzzi and Lynch, 1999]{Balduzzi1999Transaction}
Balduzzi, P. and Lynch, A.~W. (1999).
\newblock Transaction costs and predictability: Some utility cost calculations.
\newblock {\em Journal of Financial Economics}, 52(1):47--78.

\bibitem[Barroso and Saxena, 2022]{Barroso2022Lest}
Barroso, P. and Saxena, K. (2022).
\newblock Lest we forget: Learn from out-of-sample forecast errors when
  optimizing portfolios.
\newblock {\em The Review of Financial Studies}, 35(3):1222--1278.

\bibitem[Bellman, 1957]{Bellman1957Dynamic}
Bellman, R.~E. (1957).
\newblock {\em Dynamic Programming}.
\newblock Princeton University Press.

\bibitem[Bender and Thuan, 2023]{Bender2023Entropy}
Bender, C. and Thuan, N.~T. (2023).
\newblock Entropy-regularized mean-variance portfolio optimization with jumps.
\newblock https://doi.org/10.48550/arXiv.2312.13409.

\bibitem[Best and Grauer, 1991]{Best1991sensitivity}
Best, M.~J. and Grauer, R.~R. (1991).
\newblock On the sensitivity of mean-variance-efficient portfolios to changes
  in asset means: some analytical and computational results.
\newblock {\em The Review of Financial Studies}, 4(2):315–342.

\bibitem[Bj{\"o}rk et~al., 2014]{Bjoerk2014Mean}
Bj{\"o}rk, T., Murgoci, A., and Zhou, X.~Y. (2014).
\newblock Mean-variance portfolio optimization with state-dependent risk
  aversion.
\newblock {\em Mathematical Finance}, 24(1):1--24.

\bibitem[Broadie, 1993]{Broadie1993Computing}
Broadie, M. (1993).
\newblock Computing efficient frontiers using estimated parameters.
\newblock {\em Annals of Operations Research}, 45:21--58.

\bibitem[Campbell et~al., 1996]{Campbell1996Econometrics}
Campbell, J.~Y., Lo, A.~W., and MacKinlay, A.~C. (1996).
\newblock {\em The Econometrics of Financial Markets}.
\newblock Princeton University Press.

\bibitem[Candelon et~al., 2012]{Candelon2012Sampling}
Candelon, B., Hurlin, C., and Tokpavi, S. (2012).
\newblock Sampling error and double shrinkage estimation of minimum variance
  portfolios.
\newblock {\em Journal of Empirical Finance}, 19:511--527.

\bibitem[Coache and Jaimungal, 2024]{Coache2024Reinforcement}
Coache, A. and Jaimungal, S. (2024).
\newblock Reinforcement learning with dynamic convex risk measures.
\newblock {\em Mathematical Finance}, 34(2):557--587.

\bibitem[Cover and Thomas, 1991]{Cover1991Elements}
Cover, T. and Thomas, J. (1991).
\newblock {\em Elements of Information Theory}.
\newblock Wiley.

\bibitem[Czichowsky, 2013]{Czichowsky2013Time}
Czichowsky, C. (2013).
\newblock Time-consistent mean-variance portfolio selection in discrete and
  continuous time.
\newblock {\em Finance and Stochastics}, 17(2):227--271.
\newblock
  doi:\href{https://doi.org/10.1007/s00780-012-0189-9}{10.1007/s00780-012-0189-9}.

\bibitem[Dai et~al., 2023a]{Dai2023Learninga}
Dai, M., Dong, Y., and Jia, Y. (2023a).
\newblock Learning equilibrium mean-variance strategy.
\newblock {\em Mathematical Finance}, 33:1166–1212.

\bibitem[Dai et~al., 2023b]{Dai2023Learning}
Dai, M., Dong, Y., Jia, Y., and Zhou, X.~Y. (2023b).
\newblock Learning {M}erton's strategies in an incomplete market: Recursive
  entropy regularization and biased {G}aussian exploration.
\newblock https://doi.org/10.48550/arXiv.2312.11797.

\bibitem[DeMiguel et~al., 2007]{DeMiguel2007Optimal}
DeMiguel, V., Garlappi, L., and Uppal, R. (2007).
\newblock Optimal versus naive diversification: How inefficient is the 1/{N}
  portfolio strategy?
\newblock {\em The Review of Financial Studies}, 22(5):1915--1953.

\bibitem[Duarte et~al., 2024]{Duarte2024Machine}
Duarte, V., Duarte, D., and Silva, D.~H. (2024).
\newblock Machine learning for continuous-time finance.
\newblock {\em The Review of Financial Studies}, 37(11):3217--3271.

\bibitem[Fischer, 2018]{Fischer2018Reinforcement}
Fischer, T.~G. (2018).
\newblock Reinforcement learning in financial markets - a survey.
\newblock {\em FAU Discussion Papers in Economics}, 12(1):1--46.

\bibitem[Forsyth, 2020]{Forsyth2020Multiperiod}
Forsyth, P.~A. (2020).
\newblock Multiperiod mean conditional value at risk asset allocation: Is it
  advantageous to be time consistent?
\newblock {\em SIAM Journal on Financial Mathematics}, 11(2):358--384.

\bibitem[Guo et~al., 2022]{Guo2022Entropy}
Guo, X., Xu, R., and Zariphopoulou, T. (2022).
\newblock Entropy regularization for mean field games with learning.
\newblock {\em Mathematics of Operations research}, 47(4):3239--3260.

\bibitem[Haarnoja et~al., 2018a]{Haarnoja2018Softa}
Haarnoja, T., Zhou, A., Abbeel, P., and Levine, S. (2018a).
\newblock Soft actor-critic: Off-policy maximum entropy deep reinforcement
  learning with a stochastic actor.
\newblock In {\em International Conference on Machine Learning}, pages
  1861--1870, Stockholmsm{\"a}ssan, Stockholm SWEDEN.

\bibitem[Haarnoja et~al., 2018b]{Haarnoja2018Softb}
Haarnoja, T., Zhou, A., Hartikainen, K., Tucker, G., Ha, S., Tan, J., Kumar,
  V., Zhu, H., Gupta, A., and Abbeel, P. (2018b).
\newblock Soft actor-critic algorithms and applications.
\newblock https://doi.org/10.48550/arXiv.1812.05905.

\bibitem[Hiraki and Sun, 2022]{Hiraki2022toolkit}
Hiraki, K. and Sun, C. (2022).
\newblock A toolkit for exploiting contemporaneous stock correlations.
\newblock {\em Journal of Empirical Finance}, 65:99--124.

\bibitem[Hutchinson et~al., 1994]{Hutchinson1994nonparametric}
Hutchinson, J.~M., Lo, A.~W., and Poggio, T. (1994).
\newblock A nonparametric approach to pricing and hedging derivative securities
  via learning networks.
\newblock {\em The Journal of Finance}, 49(3):851--889.

\bibitem[Jiang et~al., 2022]{Jiang2022Reinforcement}
Jiang, R., Saunders, D., and Weng, C. (2022).
\newblock The reinforcement learning kelly strategy.
\newblock {\em Quantitative Finance}, 22(8):1445--1464.

\bibitem[Jobson and Korkie, 1981]{Jobson1981Putting}
Jobson, J. and Korkie, B. (1981).
\newblock Putting markowitz theory to work.
\newblock {\em Journal of Portfolio Management}, 7:70--74.

\bibitem[Jorion, 1986]{Jorion1986Bayes}
Jorion, P. (1986).
\newblock Bayes-stein estimation for portfolio analysis.
\newblock {\em The Journal of Financial and Quantitative Analysis,},
  21(3):279--292.

\bibitem[Kirby and Ostdiek, 2012]{Kirby2012Its}
Kirby, C. and Ostdiek, B. (2012).
\newblock It's all in the timing: Simple active portfolio strategies that
  outperform na{\"i}ve diversification.
\newblock {\em Journal of Financial and Quantitative Analysis}, 47(2):437--467.

\bibitem[Kot, 2014]{Kot2014First}
Kot, M. (2014).
\newblock {\em A First Course in the Calculus of Variations}, volume~72.
\newblock American Mathematical Society.

\bibitem[Kydland and Prescott, 1982]{Kydland1982Time}
Kydland, F.~E. and Prescott, E.~C. (1982).
\newblock Time to build and aggregate fluctuations.
\newblock {\em Econometrica}, 50:1345--1370.

\bibitem[Ledoit and Wolf, 2003]{Ledoit2003Improved}
Ledoit, O. and Wolf, M. (2003).
\newblock Improved estimation of the covariance matrix of stock returns with an
  application to portfolio selection.
\newblock {\em Journal of Empirical Finance}, 10(5):603--621.

\bibitem[Ledoit and Wolf, 2004a]{Ledoit2004Honey}
Ledoit, O. and Wolf, M. (2004a).
\newblock Honey, i shrunk the sample covariance matrix.
\newblock {\em The Journal of Portfolio Management Summer}, 30(4):110--119.

\bibitem[Ledoit and Wolf, 2004b]{Ledoit2004well}
Ledoit, O. and Wolf, M. (2004b).
\newblock A well-conditioned estimator for large-dimensional covariance
  matrices.
\newblock {\em Journal of Multivariate Analysis}, 88(2):365--411.

\bibitem[Ledoit and Wolf, 2017]{Ledoit2017Nonlinear}
Ledoit, O. and Wolf, M. (2017).
\newblock Nonlinear shrinkage of the covariance matrix for portfolio selection:
  Markowitz meets goldilocks.
\newblock {\em The Review of Financial Studies}, 30(12):4349--4388.

\bibitem[Li and Ng, 2000]{Li2000Optimal}
Li, D. and Ng, W.~L. (2000).
\newblock Optimal dynamic portfolio selection: Multiperiod mean‐variance
  formulation.
\newblock {\em Mathematical Finance}, 10(3):387--406.

\bibitem[Lian and Chen, 2019]{Lian2019Portfolio}
Lian, Y.-M. and Chen, J.-H. (2019).
\newblock Portfolio selection in a multi-asset, incomplete-market economy.
\newblock {\em The Quarterly Review of Economics and Finance}, 71:228--238.

\bibitem[Liberzon, 2012]{Liberzon2012Calculus}
Liberzon, D. (2012).
\newblock {\em Calculus of Variations and Optimal Control Theory}.
\newblock Princeton University Press, Princeton.

\bibitem[Markowitz, 1952]{Markowitz1952Portfolio}
Markowitz, H. (1952).
\newblock Portfolio selection.
\newblock {\em The Journal of Finance}, 7(1):77--91.

\bibitem[Markowitz, 1956]{Markowitz1956Optimization}
Markowitz, H. (1956).
\newblock The optimization of a quadratic function subject to linear
  constraints.
\newblock {\em Naval Research Logistics Quarterly}, 3:111–133.

\bibitem[Merton, 1969]{Merton1969Lifetime}
Merton, R.~C. (1969).
\newblock Lifetime portfolio selection under uncertainty: The continuous-time
  case.
\newblock {\em The Review of Economics and Statistics}, 51(3):247--257.

\bibitem[Merton, 1972]{Merton1972Analytical}
Merton, R.~C. (1972).
\newblock An analytical derivation of the efficient portfolio frontier.
\newblock {\em Journal of Financial and Quantitative Analysis}, page
  1851–1872.

\bibitem[Mnih et~al., 2016]{Mnih2016Asynchronous}
Mnih, V., Badia, A.~P., Mirza, M., Graves, A., Lillicrap, T., Harley, T.,
  Silver, D., and Kavukcuoglu, K. (2016).
\newblock Asynchronous methods for deep reinforcement learning.
\newblock In {\em International Conference on Machine Learning}, pages
  1928--1937, New York.

\bibitem[Nachum et~al., 2017]{Nachum2017Improving}
Nachum, O., Norouzi, M., and Schuurmans, D. (2017).
\newblock Improving policy gradient by exploring under-appreciated rewards.
\newblock In {\em International Conference on Learning Representations}, Palais
  des Congr{\`e}s Neptune, Toulon, Fr.

\bibitem[{\O}ksendal, 2010]{Oksendal2010Stochastic}
{\O}ksendal, B. (2010).
\newblock {\em Stochastic Differential Equations: An Introduction with
  Applications}.
\newblock Universitext. Springer, 6th ed edition.

\bibitem[Schweizer, 2010]{Schweizer2010Mean}
Schweizer, M. (2010).
\newblock Mean-variance hedging.
\newblock {\em Encyclopedia of Quantitative Finance}, pages 1177--1180.

\bibitem[Sharpe and William, 1994]{Sharpe1994The}
Sharpe and William, F. (1994).
\newblock The sharpe ratio.
\newblock {\em Journal of Portfolio Management}, 21(1):49--58.

\bibitem[Shi et~al., 2020]{Shi2020Improving}
Shi, F., Shu, L., Yang, A., and He, F. (2020).
\newblock Improving minimum-variance portfolios by alleviating overdispersion
  of eigenvalues.
\newblock {\em Journal of Financial and Quantitative Analysis},
  55(8):2700--2731.

\bibitem[Sutton and Barto, 2018]{Sutton2018Reinforcement}
Sutton, R.~S. and Barto, A.~G. (2018).
\newblock {\em Reinforcement learning: An introduction}.
\newblock MIT press.

\bibitem[Vigna, 2020]{Vigna2020time}
Vigna, E. (2020).
\newblock On time consistency for mean-variance portfolio selection.
\newblock {\em International Journal of Theoretical and Applied Finance},
  23(06):2050042.

\bibitem[Wang et~al., 2019]{Wang2019Exploration}
Wang, H., Zariphopoulou, T., and Zhou, X. (2019).
\newblock Exploration versus exploitation in reinforcement learning: A
  stochastic control approach.
\newblock {\em SSRN Electronic Journal}, pages 1--33.
\newblock
  doi:\href{https://doi.org/10.2139/ssrn.3316387}{10.2139/ssrn.3316387}.

\bibitem[Wang and Zhou, 2020]{Wang2020Continuous}
Wang, H. and Zhou, X.~Y. (2020).
\newblock Continuous-time mean-variance portfolio selection: A reinforcement
  learning framework.
\newblock {\em Mathematical Finance}, 30(4):1273--1308.
\newblock doi:\href{https://doi.org/10.1111/mafi.12281}{10.1111/mafi.12281}.

\bibitem[Wang and Forsyth, 2010]{Wang2010Numerical}
Wang, J. and Forsyth, P. (2010).
\newblock Numerical solution of the {H}amilton-{J}acobi-{B}ellman formulation
  for continuous time mean variance asset allocation.
\newblock {\em Journal of Economic Dynamics and Control}, 34(2):207--230.

\bibitem[Wang and Xie, 2024]{Wang2024Computing}
Wang, Y. and Xie, H. (2024).
\newblock Computing multi-eigenpairs of high-dimensional eigenvalue problems
  using tensor neural networks.
\newblock {\em Journal of Computational Physics}, 506:112928.

\bibitem[Won et~al., 2013]{Won2013Condition}
Won, J.-H., Lim, J., Kim, S.-J., and Rajaratnam, B. (2013).
\newblock Condition-number-regularized covariance estimation.
\newblock {\em Journal of the Royal Statistical Society Series B: Statistical
  Methodology}, 75(3):427--450.

\bibitem[Zhou and Li, 2000]{Zhou2000Continuous}
Zhou, X.~Y. and Li, D. (2000).
\newblock Continuous-time mean-variance portfolio selection: A stochastic {LQ}
  framework.
\newblock {\em Applied Mathematics and Optimization}, 42(1):19--33.

\bibitem[Zhu et~al., 2021]{Zhu2021Optimal}
Zhu, D.-M., Gu, J.-W., Yu, F.-H., Siu, T.-K., and Ching, W.-K. (2021).
\newblock Optimal pairs trading with dynamic mean-variance objective.
\newblock {\em Mathematical Methods of Operations Research}, 94(1):145--168.

\end{thebibliography}

\newpage
\appendix

\section{The Proof of Theorem \ref{th:K_k}}\label{app:K_k}

When model parameters $\mu$ and $\Sigma$ are time-independent, the average (current) profitability of $n$ risky assets remains constant, i.e., 
	$K(t,T)=A(t)=K$.
Similarly, the average (current) profitability of the $i$-th risky asset is also constant, i.e., 
\begin{align*}
	K^{(i)}(t,T)=A^{(i)}(t)=K^{(i)}, \qquad i=1,\dots,n.
\end{align*}
According to the definition of average profitability of $n$ risky assets in \eqref{equ:K}, we have
\begin{equation*}
	K=(\mu-r)^\top\Sigma^{-1}(\mu-r)
	=(\mu-r)^\top (D L D^\top)^{-1} (\mu-r).
\end{equation*}
On the other hand, we have 
	$\mu^{(i)}-r=\sqrt{K^{(i)}}\sigma^{(i)}$
then the Equation \eqref{equ:K_k} in Theorem \ref{th:K_k} can be obtained.

\section{The Proof of Lemma \ref{le:rl_aux}}\label{app:rl_aux}

We proof Lemma \ref{le:rl_aux} by contradiction.
We assume that $\{P^*(t,\cdot)\}_{0\leqslant t\leqslant T}$, is the optimal exploratory portfolio selection in problem \eqref{model_rl} but not the optimal one in auxiliary problem \eqref{model_rl_aux}.
Then, there exists $\{P(t,\theta)\}_{0\leqslant t\leqslant T}$ such that
\begin{equation}\label{equ:con}
\begin{split}
	{\rm E}\Big(-\gamma (\widetilde{W}_T^*)^2+\tau \widetilde{W}_T^*+\lambda \mathcal{H}(P^*(\cdot,\cdot))\Big)
	<{\rm E}\Big(-\gamma \widetilde{W}_T^2+\tau \widetilde{W}_T+\lambda \mathcal{H}(P(\cdot,\cdot))\Big).
\end{split}
\end{equation}
Next, we will proof that the objective function in \eqref{model_rl_aux} with $\{P(t,\theta)\}_{0\leqslant t\leqslant T}$ will become larger than that with $P^*(t,\theta)_{0\leqslant t\leqslant T}$.

It is observed that
\begin{align*}
	f(x,y)=-\gamma x+\gamma y^2+y, \quad \mbox{ with }\gamma>0,
\end{align*}
is a convex function, i.e., 
\begin{align*}
  f(x,y)\geqslant f(x_0,y_0)+f_x(x_0,y_0)\cdot(x-x_0)+f_y(x_0,y_0)\cdot(y-y_0).
\end{align*}
Setting $x={\rm E}(\widetilde{W}_T^2), y={\rm E}(\widetilde{W}_T), x_0={\rm E}(\widetilde{W}_T^*)^2, y_0={\rm E}(\widetilde{W}_T^*)$, 
we have 
\begin{equation*}
\begin{split}
	&f({\rm E}(\widetilde{W}_T^2),{\rm E}(\widetilde{W}_T))\\
  	\geqslant &f({\rm E}((\widetilde{W}_T^*)^2),{\rm E}(\widetilde{W}_T^*))-\gamma({\rm E}(\widetilde{W}_T^2)-{\rm E}((\widetilde{W}_T^*)^2))+(2\gamma {\rm E}(\widetilde{W}_T^*)+1)({\rm E}(\widetilde{W}_T)-{\rm E}(\widetilde{W}_T^*)).
\end{split}
\end{equation*}
Because of $f({\rm E}(\widetilde{W}_T^2),{\rm E}(\widetilde{W}_T))={\rm E}(\widetilde{W}_T)-\gamma {\rm Var}(\widetilde{W}_T)$,
we have 
\begin{equation*}
\begin{split}
  	&{\rm E}(\widetilde{W}_T)-\gamma {\rm Var}(\widetilde{W}_T) \\
	\geqslant &{\rm E}(\widetilde{W}_T^*)-\gamma {\rm Var}(\widetilde{W}_T^*)-\gamma {\rm E}(\widetilde{W}_T^2)+\gamma {\rm E}((\widetilde{W}_T^*)^2)
	+\tau {\rm E}(\widetilde{W}_T)-\tau {\rm E}(\widetilde{W}_T^*)
\end{split}
\end{equation*}
with $\tau=2\gamma {\rm E}(\widetilde{W}_T^*)+1$.
Thus, 
we can derive that
\begin{align*}
  	&{\rm E}(\widetilde{W}_T)-\gamma {\rm Var}(\widetilde{W}_T)
	+\lambda \mathcal{H}(P(\cdot,\cdot))\\
	\geqslant &{\rm E}(\widetilde{W}_T^*)-\gamma {\rm Var}(\widetilde{W}_T^*)-\gamma {\rm E}(\widetilde{W}_T^2)+\tau {\rm E}(\widetilde{W}_T)
	+\gamma {\rm E}((\widetilde{W}_T^*)^2)-\tau {\rm E}(\widetilde{W}_T^*)
	+\lambda \mathcal{H}(P(\cdot,\cdot))\\
  	> &{\rm E}(\widetilde{W}_T^*)-\gamma {\rm Var}(\widetilde{W}_T^*)
	+\lambda \mathcal{H}(P^*(\cdot,\cdot))
\end{align*}
which is contradictory to our assumptions that $P^*(t,\theta)$ is the optimal exploratory portfolio selection in problem \eqref{model_rl}.

\section{The Proof of Theorem \ref{th:v_theta_n}}\label{app:v_theta_n}

The derivation is divided into two parts:

1.
We first apply the high dimensional Euler-Lagrange theorem \citep{Kot2014First} to HJB equation \eqref{equ:HJB},
and derive the relationship between the optimal exploratory portfolio selection and parameter $\tau$.

According to HJB equation \eqref{equ:HJB}, the optimal exploratory portfolio selection $P^*(t,\theta)$ can be obtained by solving a constrained optimization problem:
\begin{align*}
  &\max_{P(t,\cdot)} \int_{\mathbb{R}^n} \Big(
  -\lambda \ln{P(t,\theta)}
  +\dfrac{\partial V}{\partial w}(t,w)\theta^\top (\mu-r)
  +\dfrac{1}{2}\dfrac{\partial^2 V}{\partial w^2}(t,w)\theta^\top\Sigma\theta\Big)P(t,\theta) d\theta\\
	&s.t. ~\int_{\mathbb{R}^n}P(t,\theta) d\theta=1.
\end{align*}
Thus, $P^*(t,\theta)$ satisfies
\begin{align*}
  -\lambda \ln{P^*(t,\theta)}
  +\dfrac{\partial V}{\partial w}(t,w)\theta^\top (\mu-r)
  +\dfrac{1}{2}\dfrac{\partial^2 V}{\partial w^2}(t,w)\theta^\top\Sigma\theta -k
  -\lambda =0,
\end{align*}
where $k$ is the Lagrange multiplier.
And, it follows that
\begin{align*}
	P^*(t,\theta)
	=\dfrac{{\rm exp}\Big(\frac{1}{\lambda}\Big(\frac{\partial V}{\partial w}(t,w)\theta^\top (\mu-r) +\frac{1}{2}\frac{\partial^2 V}{\partial w^2}(t,w)\theta^\top\Sigma\theta\Big)\Big)}
	{\int_{\mathbb{R}^n}{\rm exp}\Big(\frac{1}{\lambda }\Big(\frac{\partial V}{\partial w}(t,w)\theta^\top (\mu-r) +\frac{1}{2}\frac{\partial^2 V}{\partial w^2}(t,w)\theta^\top\Sigma\theta\Big)\Big) d\theta}.
\end{align*}

We conjecture the value function in the form $V(t,w)=-I(t)w^2+H(t)w+G(t)$.
Then, $P^*(t,\theta)$ is the exploratory portfolio selection with multivariate normal distribution 
\begin{align}\label{equ:multi_normal_distribution}
	\mathcal{N}\Big((\dfrac{H(t)}{2I(t)}-w)\Sigma^{-1}(\mu-r),\dfrac{\lambda}{2I(t)}\Sigma^{-1}\Big).
\end{align}
%
%
By substituting the above value function and exploratory portfolio selection back into the HJB equation \eqref{equ:HJB}, it is obtained that $I(t),H(t),G(t)$ should satisfy the ordinary differential equations
\begin{equation*}
\begin{cases}
	I^{'}(t)=-I(t)A(t)\\
	H^{'}(t)=H(t)A(t)\\
	G^{'}(t)=-\dfrac{H^2(t)}{4I(t)}A(t)-\dfrac{\lambda n}{2}\ln{\dfrac{\pi}{\lambda}}+\dfrac{\lambda n}{2}\ln(I(t))+\dfrac{\lambda}{2}\ln(|\Sigma|)
\end{cases}
\end{equation*}
with boundary conditions $I(T)=\gamma,H(T)=\tau,G(T)=0$.
Then, we calculate that
\begin{equation}\label{equ:I_H_G}
\begin{cases}
	I(t)=\gamma e^{\int_t^T -A(s) ds}\\
	H(t)=\tau e^{\int_t^T -A(s) ds}\\
	G(t)
	=\dfrac{\tau^2}{4\gamma}(1-e^{\int_t^T -A(s) ds})
	+\dfrac{\lambda n}{2}\int_t^T [\ln(\dfrac{\pi\lambda}{\gamma})
	-\dfrac{1}{n}\ln(|\Sigma|)+\int_s^T A(r)dr]ds,
\end{cases}
\end{equation}
and $H(t)$ is related to parameter $\tau$.

2.
Next, we focus on the calculation of $\tau$ with the condition $\tau=1+2\gamma {\rm E}\Big(\widetilde{W}_T^*\Big)$.

Under the exploratory portfolio selection \eqref{equ:multi_normal_distribution}, the wealth process \eqref{equ:wealth_RL} evolves as
\begin{align*}
	d\widetilde{W}_s^*
	=-\dfrac{(H(s)+2I(s)\widetilde{W}_s^*)}{2I(s)}A(s) \cdot ds
	+\sqrt{ -\dfrac{\lambda n}{2I(s)}+\dfrac{(H(s)+2I(s)\widetilde{W}_s^*)^2}{4I^2(s)}A(s)}\cdot d\widetilde{B}_s.
\end{align*}
Taking expectations on both sides of the above equation, we conclude that ${\rm E}(\widetilde{W}_s)$ satisfies the nonhomogeneous linear ordinary differential equation
\begin{align*}
	d{\rm E}(\widetilde{W}_s^*)
	={\rm E}(d\widetilde{W}_s^*)=-\dfrac{(H(s)+2I(s){\rm E}(\widetilde{W}_s^*))}{2I(s)}A(s)\cdot ds
	=\Big(-{\rm E}(\widetilde{W}_s^*)A(s)+\dfrac{\tau}{2\gamma}A(s)\Big)\cdot ds
\end{align*}
with the initial condition ${\rm E}(\widetilde{W}_0^*)=w^o$.
Thus, ${\rm E}(\widetilde{W}_s^*)$ can be expressed as
\begin{equation*}
	{\rm E}(\widetilde{W}_s^*)=e^{\int_0^s -A(t)dt}\Big( \dfrac{\tau}{2\gamma}\int_0^s A(t)e^{\int_0^t A(k)dk }dt+w^o \Big).
\end{equation*}
And, we have
\begin{equation}\label{equ:E_WT}
\begin{split}
	{\rm E}(\widetilde{W}_T^*)
	&=e^{\int_0^T -A(t)dt}\Big( \dfrac{\tau}{2\gamma}\int_0^T A(t)e^{\int_0^t A(k)dk }dt+w^o \Big)\\
	&=\dfrac{\tau}{2\gamma}
	(1-e^{\int_0^T -A(t)dt})
	+w^o e^{\int_0^T -A(t)dt}
\end{split}
\end{equation}

Substituting \eqref{equ:E_WT} back into the condition $\tau=1+2\gamma {\rm E}\Big(\widetilde{W}_T^*\Big)$, $\tau$ can be calculated as
\begin{align*}
  	\tau=e^{K(0,T)\cdot T}+2\gamma w^o
\end{align*}
with $K(0,T)$ defined in \eqref{equ:K}.
Thus, according to the above expression of $\tau$, the optimal exploratory portfolio selection \eqref{equ:multi_normal_distribution} is Gaussian with
\begin{align*}
	\mathcal{N}\Big(-(w-\dfrac{e^{K(0,T)\cdot T}+2\gamma w^o}{2\gamma})\Sigma^{-1}(\mu-r),\dfrac{\lambda}{2}\dfrac{e^{K(t,T)\cdot (T-t)}}{\gamma}\Sigma^{-1}\Big),
\end{align*}
and the corresponding value function can be expressed as
\begin{align*}
	V^*(t,w)
	&=-\gamma e^{-K(t,T)\cdot (T-t)}(w-\dfrac{e^{K(0,T)\cdot T}+2\gamma w^o}{2\gamma})^2
	+\gamma(\dfrac{e^{K(0,T)\cdot T}+2\gamma w^o}{2\gamma})^2\\
	&+\dfrac{\lambda n}{2}\int_t^T [\ln(\dfrac{\pi\lambda}{\gamma})
	-\dfrac{1}{n}\ln(|\Sigma|)+K(s,T)\cdot(T-s)] ds.
\end{align*}

\section{The Proof of Lemma \ref{le:policy_evaluation}}\label{app:policy_evaluation}

Consider the exploratory portfolio selection with probability density function
\begin{equation*}
	P(t,\cdot)=\mathcal{N}\Big((a_0-w){\bm a_1},e^{a_2}{\bm A_3}\Big), 
	\qquad\forall t\in[0,T].
\end{equation*}

(i)
The exploratory wealth process \eqref{equ:wealth_RL} becomes
\begin{equation*}
	d\widetilde{W}_t 
	= (a_0-\widetilde{W}_t){\bm a_1}^\top (\mu-r) \cdot dt 
	+\sqrt{(a_0-w)^2{\bm a_1}^\top\Sigma{\bm a_1}+{\rm tr}(e^{a_2}{\bm A_3}\Sigma)} \cdot d\widetilde{B}_t
\end{equation*}
with initial wealth $w^o$.
Taking expectations on both sides of the above equation, we conclude that ${\rm E}(\widetilde{W}_t)$ satisfies the nonhomogeneous linear ordinary differential equation
with the initial condition ${\rm E}(\widetilde{W}_0)=w^o$.
Thus, the expectation of terminal wealth is calculated as
\begin{equation*}
\begin{split}
	&{\rm E}(\widetilde{W}_T)
	=e^{\int_0^T -{\bm a_1}^\top (\mu-r)ds}
	\Big(\int_0^T a_0{\bm a_1}^\top(\mu-r)e^{\int_0^s {\bm a_1}^\top (\mu-r)dk }ds+w^o \Big).
\end{split}
\end{equation*}

(ii)
According to the Feynman-Kac formulate \citep{Oksendal2010Stochastic}, the value function $V^P(t,w)$ satisfies
\begin{equation}\label{equ:F_K}
\begin{split}
	&\dfrac{\partial V^{P}}{\partial t}(t,w)
	+\lambda h(P(t,\cdot))\\
	&+\dfrac{\partial V^{P}}{\partial w}(t,w)\int_{\mathbb{R}^n}\theta^\top (\mu-r) P(t,\theta) d\theta
	+\dfrac{1}{2}\dfrac{\partial^2 V^{P}}{\partial w^2}(t,w)\int_{\mathbb{R}^n}\theta^\top\Sigma\theta P(t,\theta) d\theta
	=0.
\end{split}
\end{equation}
When $P(t,\cdot)=\mathcal{N}\Big((a_0-w){\bm a_1},e^{a_2}{\bm A_3}\Big)$,
the PDE \eqref{equ:F_K} becomes
\begin{equation}\label{equ:FK}
\begin{split}
	&\dfrac{\partial V^{P}}{\partial t}(t,w)
	+\dfrac{\lambda n}{2}\ln{(2\pi e)}+\dfrac{\lambda}{2}\ln{|e^{a_2}{\bm A_3}|}\\
	&+\dfrac{\partial V^{P}}{\partial w}(t,w) (a_0-w){\bm a_1}^{\top}(\mu-r)
	+\dfrac{1}{2}\dfrac{\partial^2 V^{P}}{\partial w^2}(t,w)((a_0-w)^2{\bm a_1}^{\top}\Sigma{\bm a_1}+{\rm tr}(e^{a_2}{\bm A_3}\Sigma))
	=0.
\end{split}
\end{equation}
We conjecture the value function in the form $V^P(t,w)=-I^P(t)w^2+H^P(t)w+G^P(t)$.
The coefficients of the quadratic, primary and constant terms of $w$ in equation \eqref{equ:FK} are all zero,
i.e.,
\begin{equation*}
\begin{cases}
	I^{'}(t)-2I(t){\bm a_1}^{\top}(\mu-r)+I(t){\bm a_1}^{\top}\Sigma{\bm a_1}=0\\
	H^{'}(t)-2I(t)a_0{\bm a_1}^{\top}(\mu-r)+2I(t)a_0{\bm a_1}^{\top}\Sigma{\bm a_1}-H(t){\bm a_1}^{\top}(\mu-r)=0\\
	G^{'}(t)+\dfrac{\lambda n}{2}\ln{(2\pi e)}+\dfrac{\lambda}{2}\ln{|e^{a_2}{\bm A_3}|}-I(t)a_0^2{\bm a_1}^{\top}\Sigma{\bm a_1}-I(t){\rm tr}(e^{a_2}{\bm A_3}\Sigma)+H(t)a_0{\bm a_1}^{\top}(\mu-r)=0
\end{cases}
\end{equation*}
with the terminal condition 
$I^P(T)=\gamma, H^P(T)=\tau^P, G^P(T)=0$.
Solving the above PDEs, we have
\begin{equation*}
\begin{cases}
	I(t)=\gamma e^{\int_{t}^T-\left(2{\bm a_1}^{\top}(\mu - r)-{\bm a_1}^{\top}\Sigma{\bm a_1}\right)ds}\\
	H(t)=e^{\int_{t}^T-{\bm a_1}^{\top}(\mu - r)ds}\left[\tau^P-2\gamma\int_{t}^T a_0\left({\bm a_1}^{\top}\Sigma{\bm a_1}-{\bm a_1}^{\top}(\mu - r)\right)e^{\int_{s}^T-\left({\bm a_1}^{\top}(\mu - r)-{\bm a_1}^{\top}\Sigma{\bm a_1}\right)du}ds\right]
\\
	G^P(t)
	=\int_t^T
	\Big[H^P(s)a_0{\bm a_1}^\top(\mu-r)-I^P(s)a_0^2{\bm a_1}^\top\Sigma{\bm a_1}\\
	\hspace{+2cm}
	+\dfrac{\lambda n}{2}\ln{(2\pi e)}+\dfrac{\lambda n}{2}a_2+\dfrac{\lambda}{2}\ln{|{\bm A_3}|}-I^P(s) e^{a_2}{\rm tr}(\Sigma {\bm A_3})
	\Big] ds
\end{cases}
\end{equation*}

\section{The Proof of Lemma \ref{le:policy_improvement}}\label{app:policy_improvement}

For any arbitrarily given exploratory portfolio selection $P(t,\cdot)$ and another exploratory portfolio selection $\widetilde{P}(t,\cdot)$, we first calculate the difference between $V^{\widetilde{P}}(t,W_{t})$ and $V^{P}(t,W_{t})$.
Define $\{\widetilde{W}_s\}_{t<s<T}$ to be the exploratory wealth process \eqref{equ:wealth_RL} generated with the portfolio selection $\{\widetilde{P}(s,\cdot)\}_{t<s<T}$.
Under the assumption in Lemma \ref{le:policy_improvement} that $V^P(T,w)=V^{\widetilde{P}}(T,w)$, we have
\begin{align*}
	&V^{\widetilde{P}}(t,W_{t})-V^{P}(t,W_{t})
	={\rm E}_t\Big(V^{\widetilde{P}}(T,\widetilde{W}_{T})\Big)
	-V^{P}(t,W_{t})\\
	=&{\rm E}_t\Big(V^{P}(T,\widetilde{W}_{T})\Big)
	-V^{P}(t,W_{t})
	={\rm E}_t\Big(V^{P}(T,\widetilde{W}_{T})
	-V^{P}(t,W_{t})\Big)\\
	=&
	\int_t^T
	[\dfrac{\partial V^{P}}{\partial s}(s,w)
	+\dfrac{\partial V^{P}}{\partial w}(s,w)\int_{\mathbb{R}^n}\theta^\top (\mu-r) \widetilde{P}(s,\theta) d\theta\\
	&\quad+\dfrac{1}{2}\dfrac{\partial^2 V^{P}}{\partial w^2}(s,w)\int_{\mathbb{R}^n}\theta^\top\Sigma\theta \widetilde{P}(s,\theta) d\theta]ds
	+\lambda \int_t^T h(\widetilde{P}(s,\cdot)) ds.
\end{align*}
When $\widetilde{P}(t,\theta)$ is the extremum of the optimization problem
\begin{equation}\label{equ:wide_P}
\begin{split}
  &\max_{P(t,\cdot)} \int_{\mathbb{R}^n} \Big(
  -\lambda \ln{P(t,\theta)}
  +\dfrac{\partial V^{P}}{\partial w}(t,w)\theta^\top (\mu-r)
  +\dfrac{1}{2}\dfrac{\partial^2 V^{P}}{\partial w^2}(t,w)\theta^\top\Sigma\theta\Big)P(t,\theta) d\theta,\\
	&s.t. ~\int_{\mathbb{R}^n}P(t,\theta) d\theta=1,
\end{split}
\end{equation}
$\widetilde{P}(t,\theta)$ is given by the Gaussian distribution in \eqref{equ:tilde_P}.
Then, $V^{\widetilde{P}}(t,W_{t})-V^{P}(t,W_{t})$ becomes
\begin{align*}
	V^{\widetilde{P}}(t,W_{t})-V^{P}(t,W_{t})
	=& \int_t^T
	\Big[\dfrac{\partial V^{P}}{\partial s}(s,w)
	+\dfrac{\partial V^{P}}{\partial w}(s,w)\int_{\mathbb{R}^n}\theta^\top (\mu-r) \widetilde{P}(s,\theta) d\theta\\
	&+\dfrac{1}{2}\dfrac{\partial^2 V^{P}}{\partial w^2}(s,w)\int_{\mathbb{R}^n}\theta^\top\Sigma\theta \widetilde{P}(s,\theta) d\theta
	+\lambda h(\widetilde{P}(s,\cdot)\Big] ds\\
	=& \int_t^T
	\Big[\dfrac{\partial V^{P}}{\partial s}(s,w)
	+\max_{P(s,\cdot)}\Big\{\dfrac{\partial V^{P}}{\partial w}(s,w)\int_{\mathbb{R}^n}\theta^\top (\mu-r) P(s,\theta) d\theta\\
	&+\dfrac{1}{2}\dfrac{\partial^2 V^{P}}{\partial w^2}(s,w)\int_{\mathbb{R}^n}\theta^\top\Sigma\theta P(s,\theta) d\theta
	+\lambda h(P(s,\cdot)\Big\}\Big] ds.
\end{align*}
On the other hand, according to Feynman-Kac formulate, the value function $V^P(s,w)$ satisfies
\begin{align*}
	&\dfrac{\partial V^{P}}{\partial s}(s,w)
	+\dfrac{\partial V^{P}}{\partial w}(s,w)\int_{\mathbb{R}^n}\theta^\top (\mu-r) P(s,\theta) d\theta\\
	&+\dfrac{1}{2}\dfrac{\partial^2 V^{P}}{\partial w^2}(s,w)\int_{\mathbb{R}^n}\theta^\top\Sigma\theta P(s,\theta) d\theta+\lambda  h(P(s,\cdot))=0,
\end{align*}
for $s\in[t,T]$.
Thus, we have
\begin{align*}
	V^{\widetilde{P}}(t,W_{t})-V^{P}(t,W_{t})
	\geq0.
\end{align*}

\section{The Proof of Theorem \ref{th:sac}}\label{app:sac}

For any arbitrarily given initial exploratory portfolio selection $P_0(t,\theta)$ with 
\begin{equation*}
	\mathcal{N}\Big((a_0-w){\bm a_1},e^{a_2(T-t)}{\bm A_3}\Big),
\end{equation*}
according to Lemma \ref{le:policy_evaluation}, $\tau^{P_0}$ can be given as
\begin{equation*}
	\tau^{P_0}
	=1+2\gamma e^{\int_0^T -{\bm a_1}^\top (\mu-r)ds}
	\Big(\int_0^T a_0{\bm a_1}^\top(\mu-r)e^{\int_0^s {\bm a_1}^\top (\mu-r)dk }ds+w^o \Big).
\end{equation*}
The initial value function $V^{P_0}(t,w)$ with boundary condition $V^{P_0}(T,w)=-\gamma w^2+\tau_0 w$ can be shown as
\begin{align*}
	V^{P_0}(t,w)=I^{P_0}(t)w^2+H^{P_0}(t)w+G^{P_0}(t)
\end{align*}
in which
\begin{align*}
	&I^{P_0}(t)=-\gamma e^{\int_t^T -b_1(s)-b_2(s) ds},\\
	&H^{P_0}(t)
	=\tau_0 e^{\int_t^T -b_2(s) ds}
	-2\gamma e^{\int_t^T -b_2(s) ds}
	\int_t^T a_0e^{\int_s^T -b_1(r) dr}b_1(s)ds
\end{align*}
with
$b_1(s)={\bm a_1}^\top(\mu-r)-{\bm a_1}^\top\Sigma{\bm a_1}$ and 
$b_2(s)={\bm a_1}^\top(\mu-r)$.

Then, according to Lemma \ref{le:policy_improvement}, the exploratory portfolio selection is updated into $P_1(t,\theta)$ which is the probability density function of multivariate normal distribution
\begin{align*}
	P_1(t,\theta)
	=&\mathcal{N}\Big(\dfrac{\frac{\partial V^{P_0}}{\partial w}(t,w)}{-\frac{\partial^2 V^{P_0}}{\partial w^2}(t,w)}\Sigma^{-1}(\mu-r),\dfrac{\lambda}{-\frac{\partial^2 V^{P_0}}{\partial w^2}(t,w)}\Sigma^{-1}\Big)\\
	=&\mathcal{N}\Big((\dfrac{\tau_0-2\gamma \int_t^T a_0e^{\int_s^T -b_1(r) dr}b_1(s)ds}{2\gamma e^{\int_t^T -b_1(s) ds}}-w)\Sigma^{-1}(\mu-r),
	\dfrac{\lambda}{2}\dfrac{e^{\int_t^T b_1(s)+b_2(s) ds}}{\gamma}\Sigma^{-1}\Big).
\end{align*}
Again, according to Lemma \ref{le:policy_evaluation}, $\tau^{P_1}$ can be given as
\begin{equation*}
	\tau^{P_1}
	=1+2\gamma e^{\int_0^T -A(s)ds}
	\Big(\int_0^T \dfrac{\tau^{P_0}-2\gamma \int_s^T a_0e^{\int_k^T -b_1(r) dr}b_1(k)dk}{2\gamma e^{\int_s^T -b_1(k) dk}}A(s)e^{\int_0^s A(k)dk }ds+w^o \Big).
\end{equation*}
The value function $V^{P_1}(t,w)$ with boundary condition $V^{P_1}(T,w)=-\gamma w^2+\tau^{P_1} w$ becomes
\begin{align*}
	V^{P_1}(t,w)=I^{P_1}(t)w^2+H^{P_1}(t)w+G^{P_1}(t)
\end{align*}
in which
	$I^{P_1}(t)
	=-\gamma e^{\int_t^T -A(s) ds}$ and 
	$H^{P_1}(t)
	=\tau^{P_1} e^{\int_t^T -A(s) ds}$.

Again, according to Lemma \ref{le:policy_improvement}, the exploratory portfolio selection is updated to $P_2(t,\theta)$ which is the probability density function of multivariate normal distribution
\begin{align*}
	P_2(t,\theta)
	=&\mathcal{N}\Big(\dfrac{\frac{\partial V^{P_1}}{\partial w}(t,w)}{-\frac{\partial^2 V^{P_1}}{\partial w^2}(t,w)}\Sigma^{-1}(\mu-r),\dfrac{\lambda}{-\frac{\partial^2 V^{P_1}}{\partial w^2}(t,w)}\Sigma^{-1}\Big)\\
	=&\mathcal{N}\Big((\dfrac{\tau^{P_1}}{2\gamma}-w)\Sigma^{-1}(\mu-r),\dfrac{\lambda}{2}\dfrac{e^{K(t,T)\cdot(T-t)}}{\gamma}\Sigma^{-1}\Big).
\end{align*}

Then, we will prove that, for $n\geqslant2$,
\begin{align}\label{equ:tau_iter}
	\tau^{P_n}
	=1+2\gamma e^{\int_0^T -A(s)ds}
	\Big(\int_0^T \dfrac{\tau^{P_{n-1}}}{2\gamma}A(s)e^{\int_0^s A(k)dk }ds+w^o \Big)
\end{align}
and
\begin{align}\label{equ:theta_iter}
	P_{n+1}(t,\theta)
	=\mathcal{N}\Big((\dfrac{\tau^{P_{n}}}{2\gamma}-w)\Sigma^{-1}(\mu-r),\dfrac{\lambda}{2}\dfrac{e^{K(t,T)\cdot(T-t)}}{\gamma}\Sigma^{-1}\Big).
\end{align}
In fact, for $\forall k$, if
\begin{equation*}
	P_k(t,\theta)
	=\mathcal{N}\Big((\dfrac{\tau^{P_{k-1}}}{2\gamma}-w)\Sigma^{-1}(\mu-r),\dfrac{\lambda}{2}\dfrac{e^{K(t,T)\cdot(T-t)}}{\gamma}\Sigma^{-1}\Big),
\end{equation*}
$\tau^{P_k}$ can be given as
\begin{equation*}
	\tau^{P_k}
	=1+2\gamma e^{\int_0^T -A(s)ds}
	\Big(\int_0^T \dfrac{\tau^{P_{k-1}}}{2\gamma}A(s)e^{\int_0^s A(k)dk }ds+w^o \Big).
\end{equation*}
According to Lemma \ref{le:policy_evaluation}, the value function $V^{P_k}(t,w)$ with boundary condition $V^{P_k}(T,w)=-\gamma w^2+\tau^{P_k} w$ can be shown as
\begin{align*}
	V^{P_k}(t,w)=I^{P_k}(t)w^2+H^{P_k}(t)w+G^{P_k}(t),
\end{align*}
in which
	$I^{P_k}(t)
	=-\gamma e^{\int_t^T -A(s) ds}$ and
	$H^{P_k}(t)
	=\tau^{P_k} e^{\int_t^T -A(s) ds}$.
Applying Lemma \ref{le:policy_improvement}, the $k+1$-th iteration of the exploratory portfolio selection is updated to
\begin{align*}
	P_{k+1}(t,\theta)
	=&\mathcal{N}\Big(\dfrac{\frac{\partial V^{P_k}}{\partial w}(t,w)}{-\frac{\partial^2 V^{P_k}}{\partial w^2}(t,w)}\Sigma^{-1}(\mu-r),\dfrac{\lambda}{-\frac{\partial^2 V^{P_k}}{\partial w^2}(t,w)}\Sigma^{-1}\Big)\\
	=&\mathcal{N}\Big((\dfrac{\tau^{P_k}}{2\gamma}-w)\Sigma^{-1}(\mu-r),\dfrac{\lambda}{2}\dfrac{e^{K(t,T)\cdot(T-t)}}{\gamma}\Sigma^{-1}\Big).
\end{align*}
By mathematical induction, \eqref{equ:tau_iter} and \eqref{equ:theta_iter} is obtained.

Equation \eqref{equ:tau_iter} gives the recursion sequence of $\{\tau^{P_n}\}$.
Thus,
\begin{equation*}
	\lim_{n\to+\infty}\tau^{P_n}
	=\dfrac{1+2\gamma w^o e^{\int_0^T -A(s)ds}}{e^{\int_0^T -A(s)ds}}
	=e^{\int_0^T A(s)ds}+2\gamma w^o.
\end{equation*}
In this way, we draw the conclusion that
\begin{equation*}
\begin{split}
	\lim_{n\to+\infty}P_{n}(t,\theta)
	&=\lim_{n\to+\infty}
	\mathcal{N}\Big((\dfrac{\tau^{P_n}}{2\gamma}-w)\Sigma^{-1}(\mu-r),\dfrac{\lambda}{2}\dfrac{e^{K(t,T)\cdot(T-t)}}{\gamma}\Sigma^{-1}\Big)\\
	&=\mathcal{N}\Big((\dfrac{e^{\int_0^T A(s)ds}+2\gamma w^o}{2\gamma}-w)\Sigma^{-1}(\mu-r),\dfrac{\lambda}{2}\dfrac{e^{K(t,T)\cdot(T-t)}}{\gamma}\Sigma^{-1}\Big),
\end{split}
\end{equation*}
and $\lim\limits_{n\to+\infty}V^{P_n}(t,w)$ becomes the optimal value function in \eqref{equ:V_n}.

\end{document}